\documentclass[aps,prd,reprint,superscriptaddress,nofootinbib]{revtex4-1}
\pdfoutput=1
\usepackage{amsmath}
\usepackage{amssymb}
\usepackage{amsthm}
\usepackage{graphicx}
\usepackage{subfigure}
\usepackage{ulem}
\usepackage[T1]{fontenc}
\usepackage[utf8]{inputenc}
\usepackage{color}
\usepackage[hidelinks]{hyperref}
\usepackage[multiple]{footmisc}
\usepackage{bbold}

\graphicspath{{figures/}}

\newcommand{\ket} [1] {\vert #1 \rangle}
\newcommand{\bra} [1] {\langle #1 \vert}
\newcommand{\braket}[2]{\langle #1 | #2 \rangle}
\newcommand{\ii}{\mathrm{i}}

\theoremstyle{plain}

\newtheorem{theorem}{Theorem}

\newtheorem{lemma}{Lemma}
\theoremstyle{definition}

\makeatletter
\@ifundefined{textcolor}{}
{%
 \definecolor{BLACK}{gray}{0}
 \definecolor{WHITE}{gray}{1}
 \definecolor{RED}{rgb}{1,0,0}
 \definecolor{GREEN}{rgb}{0,.4,0}
 \definecolor{BLUE}{rgb}{0,0,1}
 \definecolor{CYAN}{cmyk}{1,0,0,0}
 \definecolor{MAGENTA}{cmyk}{0,1,0,0}
 \definecolor{YELLOW}{cmyk}{0,0,1,0}
 }

\definecolor{Pr}{rgb}{0.4,0.3,0.9}

\makeatother

\begin{document}

\title{Bayesian Deep Learning on a Quantum Computer}

\author{Zhikuan Zhao}
\affiliation{Department of Computer Science, ETH Zurich, Universit\"atstrasse 6, 8092 Z\"urich, Switzerland}
\affiliation{Singapore University of Technology and Design, 8 Somapah Road, Singapore 487372}
\affiliation{Centre for Quantum Technologies, National University of Singapore, 3 Science Drive 2, Singapore 117543}

\author{Alejandro Pozas-Kerstjens}
\affiliation{ICFO-Institut de Ciencies Fotoniques, The Barcelona Institute of Science and Technology, 08860 Castelldefels (Barcelona), Spain}

\author{Patrick Rebentrost}
\affiliation{Centre for Quantum Technologies, National University of Singapore, 3 Science Drive 2, Singapore 117543}

\author{Peter Wittek}
\affiliation{Rotman School of Management, University of Toronto, M5S 3E6 Toronto, Canada}
\affiliation{Creative Destruction Lab, M5S 3E6 Toronto, Canada}
\affiliation{Vector Institute for Artificial Intelligence, M5G 1M1 Toronto, Canada}
\affiliation{Perimeter Institute for Theoretical Physics, N2L 2Y5 Waterloo, Canada}

\begin{abstract}
    Bayesian methods in machine learning, such as Gaussian processes, have great advantages compared to other techniques. In particular, they provide estimates of the uncertainty associated with a prediction.
    Extending the Bayesian approach to deep architectures has remained a major challenge.
    Recent results connected deep feedforward neural networks with Gaussian processes, allowing training without backpropagation.
    This connection enables us to leverage a quantum algorithm designed for Gaussian processes and develop a new algorithm for Bayesian deep learning on quantum computers.
    The properties of the kernel matrix in the Gaussian process ensure the efficient execution of the core component of the protocol, quantum matrix inversion, providing an at least polynomial speedup over classical algorithms. Furthermore, we demonstrate the execution of the algorithm on contemporary quantum computers and analyze its robustness with respect to realistic noise models.
\end{abstract}

\maketitle

\section{Introduction}
The Bayesian approach to machine learning provides a clear advantage over traditional techniques: namely, it provides information about the uncertainty in their predictions. But not only that, they have further advantages, including automated ways of learning structure and avoiding overfitting, a principled foundation~\cite{ghahramani2015probabilistic}, and robustness to adversarial attacks~\cite{bradshaw2017adversarial,grosse2017how}.
The Bayesian framework has been making advances in various deep architectures~\cite{blundell2015weight,gal2016dropout}.
Some recent advances made a connection between a quintessentially Bayesian model, Gaussian processes (GPs)~\cite{rasmussen2006gaussian}, and deep feedforward neural networks~\cite{lee2018deep,gmatthews2018gaussian}.

Parallel to these developments, quantum technologies have been making advances in machine learning.
A new breed of quantum neural networks is aimed at current and near-future quantum computers~\cite{verdon2017quantum,torrontegui2018universal,khoshaman2018quantum,farhi2018classification,verdon2018universal} , which is in stark contrast with attempts in the past~\cite{schuld2014quest}.
Some constraints must be observed that are unusual in classical machine learning algorithms.
In particular, the protocol must be coherent, that is, we require from a quantum machine learning algorithm that it is described by a unitary map that maps input nodes to output nodes.
While the common wisdom is that a nonlinear activation is a necessary component in neural networks, a linear, unitary mapping between the inputs and outputs actually reduces the vanishing gradient problem~\cite{arjovsky2015unitary,hyland2016learning}.
Training a hierarchical representation in a unitary fashion is also possible on classical computers~\cite{liu2017machine,stoudenmire2017learning}.
So while this constraint is unusual, it is not entirely unheard of in classical machine learning, and it is the most common setting in quantum-enhanced machine learning~\cite{biamonte2017quantum}.
Furthermore, the description of quantum mechanics uses complex numbers and some promising results in machine learning show advantages of using these over real numbers~\cite{trabelsi2017deep}.

In this paper, we exploit the connection between deep learning and Gaussian processes and rely on a quantum-enhanced protocol for the latter~\cite{zhao2015quantum} to develop new algorithms that perform quantum Bayesian training of deep neural networks.
We implement the core of the algorithm on both the Rigetti Forest~\cite{smith2016practical} and the IBM QISKit~\cite{cross2017open} software stacks, and analyze how noise affects the success of the calculations on both quantum simulators.
To run on real quantum processing units, we implement a simplified, shallow-circuit version of the protocol, and compare the outcome with the simulations.
The source code is available under an open source license\footnote{\url{https://gitlab.com/apozas/bayesian-dl-quantum/}}.

\section{Background}
The algorithm that we present makes use of two previous results, which we now briefly review: the connection between deep neural networks and Gaussian processes (Section \ref{gpdl}), and the quantum Gaussian process protocol (Section \ref{sub: QGP}).

\subsection{Gaussian processes and deep learning}
\label{gpdl}
The correspondence between Gaussian processes and a neural network with a single hidden layer is well-known~\cite{neal1994priors}. Let $z(x)\in \mathbb{R}^{d_{out}}$ denote the output with input $x\in \mathbb{R}^{d_{in}}$, with $z_i(x)$ denoting the $i^{th}$ component of the output layer.
If the weight and bias parameters are taken to be i.i.d., each $z_i$ will be a sum of i.i.d. terms.
If the hidden layer has an infinite width, the Central Limit Theorem implies that $z_i$ follows a Gaussian distribution.
Now let us consider a set of $k$ input data points, with corresponding outputs $\{z_i(x^{[1]}), z_i(x^{[2]}),\ldots z_i(x^{[k]})\}$.
Any finite collection of the set will follow a joint multivariate Gaussian distribution.
Therefore $z_i$ corresponds to a Guassian process, $z_i\sim \mathcal{GP}(\mu,K)$.
Conventionally, the parameters are chosen to have zero mean, so the mean of the GP, $\mu$, is equal to zero.
The covariance matrix $K$ is given by $K(x,x^\prime)=\mathbb{E}[z_i(x)z_i(x^\prime)]$.

The Bayesian training of the neural network then corresponds to computing the posterior distribution of the given GP model, that is, calculating the mean and variance of the predictive distribution from inverting the covariance matrix.
Choosing the GP prior amounts to the selection of the covariance function and tuning the corresponding hyperparameters.
These include the information of the neural network model class, depth, nonlinearity, and weight and bias initializations.

This argument is generalized to a deep neural network architecture in a recursive manner~\cite{lee2018deep,gmatthews2018gaussian}.
Let $z_i^l$ denote the $i^{th}$ component of the output of the $l^{th}$ layer. By induction it follows that $z_i^l\sim\mathcal{GP}(0,K^l)$.
The covariance matrix on the $l^{th}$ layer is given by \mbox{$K^l(x,x^\prime)=\mathbb{E}[z_i^l(x)z_i^l(x^\prime)]$}.
To explicitly compute $K^l(x,x^\prime)$, we need to specify the variance on the weight and bias parameters, $\sigma_w^2$ and $\sigma_b^2$, as well as the nonlinearity $\phi$.
In a single-line recursive formula, this reads as
\begin{equation}
    K^l(x,x^\prime)=\sigma_b^2+\sigma_w^2\mathbb{E}[\phi(z_i^{l-1}(x))\phi(z_i^{l-1}(x^\prime))],
    \label{eq:recursion}
\end{equation}
where $z_i^{l-1}\sim\mathcal{GP}(0,K^{l-1})$. The base case of the induction is given by \mbox{$K^0(x,x^\prime)=\sigma_b^2+\sigma_w^2\left(\frac{x.x^\prime}{d_{in}}\right)$}.

Remarkably, numerical experiments suggest that the infinite-width neural network trained with Gaussian priors outperforms finite deep neural networks trained with stochastic gradient descent in many cases \cite{lee2018deep,gmatthews2018gaussian}.

\subsection{Quantum Gaussian process algorithm}
\label{sub: QGP}
A quantum algorithm for Gaussian process regression was introduced in Ref.~\cite{zhao2015quantum}.
Given a supervised learning problem with a training dataset with input points $\{x_i\}_{i=0}^{n-1}$ and corresponding output points $\{y_i\}_{i=0}^{n-1}$, the quantum GP algorithm leverages the quantum linear system subroutine introduced in Ref.~\cite{Harrow2009a}, and computes a GP model's mean predictor,
\begin{align}
\bar{{f}_*}=k_*^T(K+\sigma_n^2I)^{-1}y
\label{eq:meanpred}
\end{align}
and variance predictor,
\begin{align}
\mathbb{V}[{f}_*]=k\left(x_*,x_*\right)-k_*^T(K+\sigma_n^2I)^{-1}k_*.
\label{eq:variancepred}
\end{align}
Here $(K+\sigma_n^2I)$ denotes the model's covariance matrix with Gaussian noise entries of variance $\sigma_n^2$, and $k_*$ denotes the row in the covariance matrix that corresponds to the target point for prediction.
The scalar $k\left(x_*,x_*\right)$ is the covariance function of the target point with itself, and takes only a constant time to compute.

Assuming a black-box access to the matrix elements of $K$, the quantum GP algorithm simulates $(K+\sigma_n^2I)$ as a Hamiltonian acting on an input state, $\ket{b}$, performs quantum phase estimation~\cite{Kitaev1995} to extract estimates of the eigenvalues of $(K+\sigma_n^2I)$, and stores them in a quantum register as a weighted superposition.
While in superposition, the stored eigenvalues are inverted and used to construct a controlled rotation on an ancillary system.
Conditioned on a final measurement result on the ancillary system, the algorithm probabilistically completes a computation for $(K+\sigma_n^2I)^{-1}\ket{b}$.
Depending on whether the aim is computing the mean predictor or the variance predictor, one chooses $\ket{b}=\ket{y}$ or $\ket{b}=\ket{k_*}$, which encodes the classical vectors $y$ or $k_*$ respectively. Finally applying a quantum inner product routine, such as those described in Refs. \cite{tacchino2018,schuld2018FH}, allows for a good estimation of the quantities $k_*^T(K+\sigma_n^2I)^{-1}y$ and $k_*^T(K+\sigma_n^2I)^{-1}k_*$, which leads to the goal of a GP regression model computation.

The quantum GP algorithm runs in $\tilde{\mathcal{O}}(\log(n))$ time when $K$ is sparse and well-conditioned.
A caveat here is that the quantum algorithm only runs in logarithmic time for sparse covariance matrices, and this could restrict the form of the non-linear function or other parameters in the network architecture.
The simulation of sparse Hamiltonians is more efficient when using quantum computers~\cite{lloyd1996universal,childs2010relationship,berry2012black}.
This can be addressed by tapering the covariance function using a compactly supported function~\cite{furrer2006covariance}; a similar methodology is also known in kernel methods~\cite{wittek2011compact}.
Furthermore, one could apply the methods in Ref.~\cite{wossnig2017quantum} to construct a $\mathcal{O}(\sqrt{n})$ time algorithm for Gaussian processes.
This should ensure at least a polynomial quantum speedup for general constructions. Subsequently to the quantum GP algorithm, a corresponding quantum method for enhancing the training and model selection of GPs was introduced in Ref.~\cite{zhao2018quantum}.

\section{Quantum Bayesian training of neural networks}\label{sec:algorithm}
Now, we leverage the previous two results to develop a way of conducting Bayesian training of deep neural networks using a Gaussian prior.

According to the connection described in Section~\ref{gpdl}, Bayesian training of a deep neural network of $L$ layers requires sampling the values of the neurons in the final layer from the Gaussian process $\mathcal{GP}(0,K^L)$, where $K^L$ can be computed in a recursive manner beginning from $K^0$ following Eq.~\eqref{eq:recursion}.
If we had classical access to the elements of the covariance matrix $K^0$, one possibility could be to classically compute $K^L$ and then resort to the simulation of the Hamiltonian evolution generated by $K^L$ to obtain the mean predictor $\bar{{f}_*}$ and variance predictor $\mathbb{V}[{f}_*]$ needed in the quantum Gaussian process algorithm of Section~\ref{sub: QGP}~\cite{zhao2015quantum}. This procedure would require simulating the Hamiltonian evolution from a classical encoding of $K^L$, which may hinder the speedup expected from the algorithm in this case.

The algorithm we propose makes use of the following observation: for the quantum Gaussian process algorithm there is no need to have a complete knowledge of the covariance matrix. In reality, one just needs to know the time evolution operator under the covariance matrix encoded as a Hamiltonian.
We propose a way of constructing such time evolution operator given access to a quantum encoding of the base case covariance matrix $K^0$, either in the form of oracular access or encoded as a density matrix of a qubit system (we discuss both possibilities later in this section).
Once the time evolution operator is simulated, our algorithm, as the quantum Gaussian process algorithm, needs sampling from only one Gaussian process, that corresponding to the last layer in the network.

A requirement of the algorithm is, as in the classical case, a functional expression of the covariance matrix in the last layer in terms of the base case $K^0$.
For general non-linear activation functions, this can only be done with numerical integration, which seems quite unreachable to implement coherently with contemporary quantum computers.
A complete quantum protocol would require a large number of qubits and at least polynomial-size quantum circuits, which remains out of reach with current technology.
However, different works showed activation functions which yield kernels and recursion relations that can be analytically calculated or approximated~\cite{cho2009kernel,daniely2016toward}.
A particularly useful special case amounts to using only the ReLU non-linear activation on every layer.
The ReLU activation function is $\phi(x)=\max(0,x)$, and has been crucial in addressing issues such as the vanishing gradient problem in deep learning~\cite{glorot2011relu}.
For this case, the $l^{th}$ layer covariance matrix has an analytical formula~\cite{lee2018deep}:
\begin{align}
K^l(x,x^\prime)=&\,\sigma_b^2+\frac{\sigma_w^2}{2\pi}\sqrt{K^{l-1}(x^\prime,x^\prime)K^{l-1}(x,x)}\notag\\
&\times\!\left[\arcsin(\theta^{l-1}_{x,x^{\prime}})\!-\!(\pi-\theta^{l-1}_{x,x^{\prime}})\arccos(\theta^{l-1}_{x,x^{\prime}})\right]
\label{Relu},
\end{align}
where $$\theta^{l}_{x,x^{\prime}}=\arccos\left(\frac{K^{l}(x,x^\prime)}{\sqrt{K^{l}(x,x)K^{l}(x^\prime,x^\prime)}}\right).$$

The non-linear functions featured in Eq.~\eqref{Relu} can be approximated by polynomial series with some convergence conditions. The factor $K^{l}(x,x)K^{l}(x^\prime,x^\prime)$ represents outer products between the two identical vectors of diagonal entries in $K^{l}$. As such, the computation of Eq.~\eqref{Relu} can be decomposed into such outer product operations combined with element-wise matrix multiplication. In Sections~\ref{sub:multilayer} and~\ref{sub: Element-wise} we provide a construction for simulating the evolution under the Hamiltonians generated by these operations on the matrix elements of a quantum state.

For the remaining discussion, we briefly introduce the mathematical formalism of quantum computing.
In particular, a ket $\ket{x}$ denotes a column vector $x\in\mathbb{C}^d$ for some dimension $d$, with norm 1.
Its complex conjugate is a bra $\bra{x}$.
A ket represents a pure quantum state.
A quantum computer essentially transforms quantum states into quantum states, and the result of the quantum computation is a quantum state with some desired properties.
The density matrix of a pure state is the outer product of ket and the corresponding bra, and it is a positive semidefinite matrix with trace 1.
For pure states, the density matrix is an equivalent way of describing a quantum state.
In addition, the density matrix allows to describe mixed quantum states, i.e.~statistical ensembles of pure states.
For the algorithm proposed here it is needed that $K^0$ is given as a real symmetric, positive semi-definite matrix, normalized by its trace in order to qualify as a quantum state~\cite{rebentrost2014quantum}.
All but the last property are satisfied by the definition of covariance matrix, and the last one can be achieved with an appropriate rescaling, equivalent to an appropriate choice of the kernel function.
For more details on quantum computations, we refer the reader to Ref.~\cite{nielsen2000quantum}.

As introduced above, the quantum algorithms used in the present work can admit two data-input models. First, we can assume efficient computability or oracular access to the matrix elements of the covariance matrix $K^0$. In this model, the quantum simulation methods of~\cite{berry2012black,berry2015hamiltonian} can be used in the quantum GP algorithm, as long as the assumptions of these methods are satisfied. Second, we can assume that the covariance matrix is presented as the quantum density matrix of a qubit system. Multiple copies of such a density matrix allow the use of a method inspired by the quantum principal component analysis algorithm~\cite{rebentrost2014quantum}. We discuss the first method for the single-layer case and the second method for the
multiple-layer case.

\subsection{Single-layer case}

Assume that we are given oracle access to the matrix elements of the base case:
$$
O_{K^0} \ket{j,k} \ket{z} \to  \ket{j,k} \ket{z \oplus K^0_{jk}},
$$
where the matrix elements are written in the notation \mbox{$K^0_{jk}=K^0(x_j,x_k)$}.
The desired kernel function of Eq.~\eqref{Relu} can be implemented by oracle queries using ancillary labeling registers with $\ket{j,j}$, $\ket{k,k}$ and $\ket{j,k}$, as well as an additional register which stores the value of a classical computation step.
This procedure can be described as follows:
\begin{equation}
O_{K^0}\ket{j,j}\ket{k,k}\ket{j,k}\ket{0} \to  \ket{j,j}\ket{k,k}\ket{j,k} \ket{0 \oplus K^1_{jk}}.
\end{equation}
With the oracle access to the elements of $K^0$, the first, and final, layer covariance matrix $K^1$ can be classically computed and simulated as a Hamiltonian used in the quantum GP algorithm.

\subsection{Multi-layer case}
\label{sub:multilayer}

In the case of multi-layer network architectures, we describe a method to simulate the $l^{th}$-layer kernel matrix as a Hamiltonian.
Our approach is inspired by the quantum principal component analysis algorithm~\cite{rebentrost2014quantum} where the density matrix $\rho$ of a quantum state is treated as a Hamiltonian and used to construct the desired controlled unitary $e^{i t\rho}$ acting on a target quantum state for a time period $t$.
This is an unusual concept for classical machine learning and classical algorithms: a high-dimensional vector becomes an operator on itself to reveal its own eigenstructure.
A throughout description of this density matrix-based Hamiltonian simulation procedure is presented in Ref.~\cite{kimmel2017hamiltonian}.
Here we will first give an overall description of the quantum method, while the detailed analysis is presented later in the paper.

In order to apply density matrix-based Hamiltonian simulation using the $l^{th}$-layer kernel, we need to incorporate methods to compute certain element-wise matrix operations between two density matrices.
It is convenient to define the following:

\begin{align*}
S_1=\sum_{ j, k} |j\rangle \langle k|\otimes |j\rangle \langle k|  \otimes  |k\rangle \langle  j |,\\
S_2=\sum_{ j, k} |j\rangle \langle j|\otimes |k\rangle \langle k|  \otimes  |k\rangle \langle  j |.
\end{align*}

With an augmented density matrix exponentiation scheme, $S_1$ computes exponential of the Hadamard product of two density matrices, while $S_2$ computes the exponential of the outer product between the diagonal entries of two density matrices.
Specifically, we have
\begin{align}
{\rm tr}_{1,2} & \{ e^{- \ii S_1 \delta} ( \rho_1 \otimes \rho_2 \otimes \sigma ) e^{ \ii S_1 \delta} \} \notag\\
&=\exp[-\ii(\rho_1 \odot \rho_2)\delta]\,\sigma\exp[\ii(\rho_1 \odot \rho_2)\delta]+ \mathcal{O}(\delta^2),
\label{element-wise product}
\end{align}
where $\rho_1 \odot \rho_2$ denotes the Hadamard product between $\rho_1$ and $\rho_2$, and ${\rm tr}_{1,2}$ denotes tracing out the first and second subsystems, respectively.
The factor $\delta$ represents a small evolution time with the operator in the exponents.
We also have
\begin{align}
{\rm tr}_{1,2} &\{ e^{- \ii S_2 \delta} ( \rho_1 \otimes \rho_2 \otimes \sigma ) e^{ \ii S_2 \delta} \} \notag\\
&=\exp[-\ii(\rho_1 \oslash \rho_2)\delta]\sigma\exp[\ii(\rho_1 \oslash \rho_2)\delta]+ \mathcal{O}(\delta^2),
\label{diagonal outer product}
\end{align}
where $\rho_1 \oslash \rho_2$ denotes taking the outer product between the diagonal entries of $\rho_1$ and $\rho_2$.
The derivation of Eqs.~\eqref{element-wise product} and \eqref{diagonal outer product} are presented later in Section~\ref{sub: Element-wise}.
Both $S_1$ and $S_2$ are sparse and thus efficiently simulable as a Hamiltonian with methods based on quantum walks~\mbox{\cite{berry2012black,berry2015hamiltonian}}.
A similar method of using a modified version of the SWAP operator combined with density matrix exponentiation scheme was used in~\cite{rebentrost2018quantum} for a quantum singular value decomposition algorithm.

In order to approximately compute the non-linear function of Eq.~\eqref{Relu}, we make use of a polynomial series in $K^0(x,x^\prime)$.
Note that due to the structure of Eq.~\eqref{Relu}, the products involved in this polynomial series are the Hadamard products denoted by $\odot$, and the diagonal outer products denoted by $\oslash$.
We will denote the polynomial in $K^0$ to the order $N(l)$ which approximates the $l^{th}$ layer kernel function as $P^{N}_{(\odot,\oslash)}(K^0)$.

We note that by using a generalized $\tilde{S}$ operator which combines the components in $S_1$ and $S_2$, one can implement a total $N$ number of $\odot$ and $\oslash$ operations in arbitrary orders.
In Section~\ref{sub: Element-wise}, we will show this simply amounts to summing over the tensor product of the projectors $|j\rangle \langle j|$, $|j\rangle \langle k|$, and $|k\rangle \langle k|$.
Similar polynomial series simulation problems were addressed in Refs.~\cite{kimmel2017hamiltonian} and~\cite{rebentrost2016quantum}, but the type of product considered was standard matrix multiplication instead of element-wise operations.

The quantum technique described above combined with using the series expansions of the non-linear functions in Eq.~\eqref{Relu} gives us a way to approximate $e^{\ii tK^l}\sigma e^{-\ii tK^l}$, where $\sigma$ is an arbitrary input state.
Hence given multiple copies of a density matrix which encodes the initial layer covariance matrix, $K^{0}$, the unitary operator, $\exp(-\ii t K^l)$ can be constructed to act on an arbitrary input state, as required by applying the quantum GP algorithm described in Section~\ref{sub: QGP}.
Note that there is a subtle but crucial difference between the single-layer and the multilayer case: while in the training of single-layer networks one needs of a quantum random access memory to perform the oracle queries of the matrix elements of $K^0$, in the multilayer case we substitute this requirement by having access to multiple copies of the quantum state encoding $K^0$.
This requirement is much more feasible given current technology since the desired state preparation can be encoded in a quantum circuit and run as many times as needed.

\subsection{Coherent element-wise operations}
\label{sub: Element-wise}
In this section we give a more formal description of the quantum method for approximately compute the polynomial $P^{N}_{(\odot,\oslash)}(K^0)$. The main results needed are well summarised by the following Lemmas~\ref{lemma: hadamard} and~\ref{lemma: outer}, and Theorem~\ref{theorem: poly}.

\begin{lemma}\label{lemma: hadamard}
Given $\mathcal{O}(t^2/\epsilon)$ copies of $d$-dimensional qubit density matrices, $\rho_1$ and $\rho_2$, let $\rho_1 \odot \rho_2$ denote the Hadamard product between $\rho_1$ and $\rho_2$.
There exists a quantum algorithm to implement the unitary $e^{-\ii \rho_1 \odot \rho_2 t}$ on a $d$-dimensional qubit input state $\sigma$,
for a time $t$ to accuracy $\epsilon$ in operator norm.
\end{lemma}

\begin{proof}
The usual $\text{SWAP}$ matrix employed in quantum principal component analysis~\cite{rebentrost2014quantum} is given by
\mbox{$S=\sum_{ j, k} |j\rangle \langle k|  \otimes  |k\rangle \langle  j |$}.
Here we take the modified $\text{SWAP}$ operator
$
S_1=\sum_{ j, k} |j\rangle \langle k|\otimes |j\rangle \langle k|  \otimes  |k\rangle \langle  j |.
$
With an arbitrary input state $\sigma$, the operation
\begin{equation}
{\rm tr}_{1,2} \{ e^{- \ii S_1 \delta} ( \rho_1 \otimes \rho_2 \otimes \sigma ) e^{ \ii S_1 \delta} \}
\label{eq:S1evo}
\end{equation}
can be efficiently performed with a small parameter $\delta$.
The symbol ${\rm tr}_{1,2}$ represents the trace over the subspaces of $\rho_{1}$ and $\rho_{2}$.
Expanding Eq.~\eqref{eq:S1evo} to $\mathcal{O}(\delta^2)$ leads to:
\begin{align}
{\rm tr}_{1,2} \{ e&^{- \ii S_1 \delta}( \rho_1 \otimes \rho_2 \otimes \sigma ) e^{ \ii S_1 \delta} \} \nonumber \\
= 1&-\ii \, {\rm tr}_{1,2} \{ S_1 ( \rho_1 \otimes \rho_2 \otimes \sigma ) \} \delta \nonumber \\
&+ \ii\, {\rm tr}_{1,2} \{ ( \rho_1 \otimes \rho_2 \otimes \sigma ) S_1 \}\delta \nonumber \\
&+\mathcal{O}(\delta^2).
\end{align}

Examining the first element linear in the parameter $\delta$ reveals
\begin{align}
{\rm tr}_{1,2}  &\{ S_1 ( \rho_1 \otimes \rho_2 \otimes \sigma )  \} \notag\\
&= {\rm tr}_{1,2} \{ \sum_{ j, k} |j\rangle \langle k|\otimes |j\rangle \langle k|  \otimes  |k\rangle \langle  j | ( \rho_1 \otimes \rho_2 \otimes \sigma )  \} \notag \\
&= \sum_{ n, m,j,k}\langle n|j\rangle \langle k|\rho_1 |n\rangle \langle m |j\rangle \langle k|  \rho_2 |m\rangle  |k\rangle \langle  j | \sigma  \notag \\
&= \sum_{ j,k} \langle k|\rho_1 |j \rangle  \langle k|  \rho_2 |j\rangle  |k\rangle \langle  j | \sigma  \notag \\
&=   (\rho_1 \odot \rho_2 )\, \sigma.
\end{align}

In the same manner we have
\begin{eqnarray}
{\rm tr}_{1,2}  \{  ( \rho_1 \otimes \rho_2 \otimes \sigma )S_1  \} &=& \sigma (\rho_1 \odot \rho_2 ).
\end{eqnarray}

Thus in summary, we have shown that
\begin{align}
{\rm tr}_{1,2} \{ e&^{- \ii S_1 \delta} ( \rho_1 \otimes \rho_2 \otimes \sigma ) e^{ \ii S_1 \delta} \} \nonumber\\
=&
\,\sigma -\ii [(\rho_1 \odot \rho_2 ) ,\sigma] \delta + \mathcal{O}(\delta^2).
\end{align}

\noindent The above is equivalent to applying the unitary $\exp[-\ii(\rho_1 \odot \rho_2)\delta]$ to $\sigma$ up to $\mathcal{O}(\delta)$:
\begin{align}
\exp[&-\ii(\rho_1 \odot \rho_2)\delta]\sigma\exp[\ii(\rho_1 \odot \rho_2)\delta]\nonumber\\
=&[\mathbb{1}-\ii(\rho_1 \odot \rho_2)\delta+O(\delta^2)]\,\sigma[\mathbb{1}+\ii(\rho_1 \odot \rho_2)\delta+O(\delta^2)]\nonumber\\
=&\sigma -\ii[(\rho_1 \odot \rho_2 ),\sigma]\delta +\mathcal{O}(\delta^2).
\end{align}

The above completes the derivation of Eq.~\eqref{element-wise product}. Note that if the small time parameter is taken to be \mbox{$\delta=\epsilon/t$}, and the above procedure is implemented $\mathcal{O}(t^2/\epsilon)$ times, the overall effect amounts to implementing the desired operation, $e^{-\ii\rho t}\sigma e^{\ii\rho t}$ up to an error $\mathcal{O} (\delta^2 t^2/\epsilon )= \mathcal{O} (\epsilon )$, while consuming $\mathcal{O} (t^2/\epsilon )$ copies of $\rho_1$ and $\rho_2$.
This concludes the proof of Lemma~\ref{lemma: hadamard}.
\end{proof}

\begin{lemma}\label{lemma: outer}
Given $\mathcal{O}(t^2/\epsilon)$ copies of $d$-dimensional qubit density matrices, $\rho_1$ and $\rho_2$, let $\rho_1 \oslash \rho_2$ denote the outer product between the diagonal entries of $\rho_1$ and $\rho_2$.
There exists a quantum algorithm to implement the unitary $e^{-\ii \rho_1 \oslash \rho_2 t}$ on a $d$-dimensional qubit input state, $\sigma$
for a time $t$ to accuracy $\epsilon$ in operator norm.
\end{lemma}

\begin{proof}
By simply re-indexing the operator $S_1$, one obtains $S_2=\sum_{ j, k}|j\rangle \langle j|\otimes |k\rangle \langle k|  \otimes  |k\rangle \langle  j |$. Analogously with the proof of Lemma~\ref{lemma: hadamard}, we have
\begin{align}
{\rm tr}_{1,2} &\{ e^{- \ii S_2 \delta} ( \rho_1 \otimes \rho_2 \otimes \sigma ) e^{ \ii S_1 \delta} \} \nonumber\\
&=
\,\sigma -\ii [(\rho_1 \oslash \rho_2 ) ,\sigma] \delta + O(\delta^2).
\end{align}

The above can be compared with
\begin{align}
\exp[-\ii&(\rho_1 \oslash \rho_2)\delta]\,\sigma\exp[\ii(\rho_1 \oslash \rho_2)\delta]\nonumber\\
=&\,\sigma -\ii[(\rho_1 \oslash \rho_2 ),\sigma]\delta + O(\delta^2).
\end{align}

The equivalence up to the linear term in $\delta$ confirms the validity of Eq.~\eqref{diagonal outer product}.
Similarly with Lemma~\ref{lemma: hadamard}, with a $\mathcal{O}(t^2/\epsilon)$ repetition consuming $\mathcal{O}(t^2/\epsilon)$ copies of $\rho_1$ and $\rho_2$, the desired $e^{-\ii\rho t}\sigma e^{\ii\rho t}$ can be implemented up to error $\epsilon$.
\end{proof}

Given the density matrix $\rho=K^0$ which encodes the base case covariance matrix, we approximate the non-linear kernel function at $l^{th}$ layer with the order $N$ polynomial, $P^{N}_{(\odot, \oslash)}(\rho)=\sum_r^N c_r \rho ^{ (\odot, \oslash)r}$.
Here the label $(\odot, \oslash)$ indicates that we work in the setting where the types of product operation involved for taking the $r^{th}$ power of $\rho$ are arbitrary combinations of Hadamard and diagonal outer products.
Now we are in the position of presenting the main theorem required to implement the kernel function at the $l^{th}$ layer.

\begin{theorem}\label{theorem: poly}
Given $\mathcal{O}(N^2 t^2/\epsilon)$ copies of the $d$-dimensional qubit density matrix $\rho$, and the order-$N$ polynomial of Hadamard and diagonal outer products, $$P^{N}_{(\odot,\oslash)}(\rho)=\sum_r^N c_r \rho ^{ (\odot, \oslash)r},$$
there exists a quantum algorithm to implement the unitary $e^{-\ii P^{N}_{(\odot, \oslash)}(\rho) t}$ on a $d$-dimensional qubit input state $\sigma$
for a time $t$ to accuracy $\epsilon$ in operator norm.
\end{theorem}
\begin{proof}
We first address how to implement the unitary $e^{-\ii \rho ^{ (\odot, \oslash)r } t}$.
Intuitively, this can be achieved by constructing a generalized $\tilde{S}$ operator with tensor product components of $|j\rangle \langle j|$, $|j\rangle \langle k|$, $|k\rangle \langle k|$ and $|k\rangle \langle j|$, corresponding to the contributing elements in the matrices in each term.
We give a recursive procedure to determine $\tilde{S}$:

In the case of $r=2$, we have already shown in Lemma~\ref{lemma: hadamard} and Lemma~\ref{lemma: outer} the desired operation can be achieved using $S_1$ and $S_2$ corresponding to the $\odot$ and $\oslash$ cases respectively.
Thus we can write the base case of the recursive procedure as
$$\tilde{S}^{(r=2)}= \sum_{j,k} T^{(2)}(j,k)\otimes |k\rangle \langle j|,$$ where $T^{(2)}(j,k)$ denotes the possible combinations of tensor products, $|j\rangle \langle k|\otimes |j\rangle \langle k|$ or $|j\rangle \langle j|\otimes |k\rangle \langle k|$.
Now consider the $r=3$ case, the additional factor of $\rho$ will come in two possible cases. If it comes as a $\odot$ product, the updated operator $\tilde{S}^{(r=3)}_\odot$ is simply given by
$$
\tilde{S}^{(r=3)}_\odot = \sum_{j,k} T^{(2)}(j,k)\otimes |j\rangle \langle k| \otimes |k\rangle \langle j|.
$$

If the additional $\rho$ comes in as a $\oslash$ product, the updated operator $\tilde{S}^{(r=3)}_\oslash$ is instead given by

$$
\tilde{S}^{(r=3)}_\oslash = \sum_{j,k} |j\rangle \langle j| \otimes |j\rangle \langle j|  \otimes |k\rangle \langle k| \otimes |k\rangle \langle j|.
$$

This can be seen by observing that the contributing elements to a $\oslash$ product are exclusively diagonal, which we use $|j\rangle \langle j|$ to pick up.
Any off-diagonal information about the previous element-wise product operations is irrelevant. In general, if we have the $r^{th}$ order $\tilde{S}$ operator given by

$$
\tilde{S}^{(r)} = \sum_{j,k} T^{(r)}(j,k)\otimes |k\rangle \langle j|,
$$
the operators $\tilde{S}^{(r+1)}_\odot$ and $\tilde{S}^{(r+1)}_\oslash$ can be generated as follows:

\begin{align}
\tilde{S}^{(r+1)}_\odot =& \sum_{j,k} T^{(r)}(j,k)\otimes |j\rangle \langle k| \otimes |k\rangle \langle j|,    \nonumber \\
\tilde{S}^{(r+1)}_\oslash =& \sum_{j,k} (|j\rangle \langle j|)^{\otimes r}  \otimes |k\rangle \langle k| \otimes |k\rangle \langle j|.
\end{align}

We have shown a recursive procedure to construct $\tilde S^{(r)}$ up to $r=N$ such that
\begin{align}
{\rm tr}_{1...r} &\{ e^{- \ii \tilde S^{(r)} \delta} ( \rho^{\otimes r} \otimes \sigma ) e^{ \ii \tilde S^{(r)} \delta} \} \notag\\
=&\exp[-\ii \rho ^{ (\odot, \oslash)r }\delta]\,\sigma\exp[\ii \rho ^{ (\odot, \oslash)r }\delta]+ \mathcal{O}(\delta^2),
\end{align}
for a small evolution $\delta$. Analogously with Lemma~\ref{lemma: hadamard} and Lemma~\ref{lemma: outer}, with a $\mathcal{O}(t^2/\epsilon)$ repetition consuming $\mathcal{O}(rt^2/\epsilon)$ copies of $\rho$, the desired $$\exp[-\ii \rho ^{ (\odot, \oslash)r }t]\,\sigma\exp[\ii \rho ^{ (\odot, \oslash)r }t]$$ can be implemented up to an $\epsilon$ error.

Finally one makes use of the Lie product formula for summing the terms in the polynomial     \cite{suzuki1992general,childs2003exponential,wiebe2010higher}:
\begin{eqnarray}
e^{\ii\delta ( A + B)+ \mathcal{O} (\delta ^2/m)}=(e^{\ii\delta A/m}e^{\ii \delta B/m})^m,
\end{eqnarray}
where $A$ and $B$ are taken to different terms in $P^{N}_{(\odot,\oslash)}(\rho)=\sum_r^N c_r \rho ^{ (\odot, \oslash)r}$, and the factors $c_r$ simply amount to multiplying the $S^{(r)}$ matrices with the respective coefficients.
The parameter $m$ can be chosen further suppress the error by repeating the entire procedure.
However, for the purpose of implementing $$e^{-\ii P^{N}_{(\odot, \oslash)}(\rho) t}\sigma e^{\ii P^{N}_{(\odot, \oslash)}(\rho) t} $$ to our desired accuracy $\epsilon$, $\mathcal{O}(N^2 t^2/\epsilon)$ copies of $\rho$ are required.
The quadratic dependency in the order of the polynomial, $N^2$, stems from implementing the unitary $\exp[-\ii \rho ^{ (\odot, \oslash)r }t]$ up to $r=N$, each consuming $\mathcal{O}(Nt^2/\epsilon)$ copies as argued before.
\end{proof}

\section{Experiments}
The central part of the algorithm described in Section~\ref{sec:algorithm} is the intricate quantum protocol of matrix inversion for computing the predictors in Eqs.~\eqref{eq:meanpred},~\eqref{eq:variancepred}.
This protocol~\cite{Harrow2009a} is probabilistic, meaning that it only succeeds conditioned on obtaining specific results after measuring specific qubits in the protocol.
Therefore, it is not assured that the protocol will succeed in a particular run, and it has to be repeatedly performed until it succeeds in obtaining the correct solution.
Moreover, computations on real quantum computers are subject to imprecisions in the gates applied to the qubits, readout errors and losses of coherence in the state of the system.

Therefore, when thinking about a realistic application of the quantum Bayesian algorithm, the important questions to ask are how experimentally feasible it is, and how far we are from running it on real quantum computers.
With this goal in mind, we have performed two sets of experiments: on the one hand, we have run simulations of the quantum matrix inversion protocol on two different quantum virtual machines with various noise models that affect real quantum computers, and analyzed their impact on the output---the final quantum state after the protocol---of the algorithm.
On the other hand, we have run scaled-down versions of the protocol on two real, state-of-the-art quantum processing units to gauge how far we are from implementations of practical relevance.

We have implemented the complete quantum matrix inversion protocol in the Rigetti Forest API using PyQuil and Grove~\cite{smith2016practical}. This implementation can perform approximate eigenvalue inversion on a Hermitian matrix of arbitrary size.
The PyQuil framework has advanced gate decomposition features and provides a way to perform arbitrary unitary operations on a multi-qubit quantum state.
Furthermore, Rigetti's classical simulator of quantum circuits (referred to as a \emph{quantum virtual machine}) provides a variety of noise models that can affect computations in real quantum architectures, allowing a detailed analysis of how noise affects accuracy and computational overhead.

In addition, we have implemented reduced, $2\times 2$ matrix inversion problems in both PyQuil---to be run in Rigetti's Quantum Processing Unit---and in IBM's QISKit software stack~\cite{cross2017open}---to be run in IBM's Quantum Experience computers---.
QISKit also provides a noisy classical simulator, of which we also make use to contrast the performance of the quantum matrix inversion algorithm run in the real QPUs against simulations with realistic noise models.

The quantum processing units employed in the experiments are IBM's 16-qubit Rueschlikon \mbox{(IBMQX5)}~\cite{wang2018ibm} and Rigetti's 8-qubit 8Q-Agave.
While the number of available physical qubits is in both cases higher than the number of qubits required for the implementation (a total of six for the $2\times 2$ reduced version), the depth of the circuit is much higher for larger matrices, and the current noise levels in the QPUs would not allow obtaining meaningful results when inverting larger examples.

\subsection{Simulations of algorithm success on a quantum virtual machine}
\begin{figure*}[t]
    \centering
    \subfigure[]{
        \includegraphics[width=0.48\textwidth]{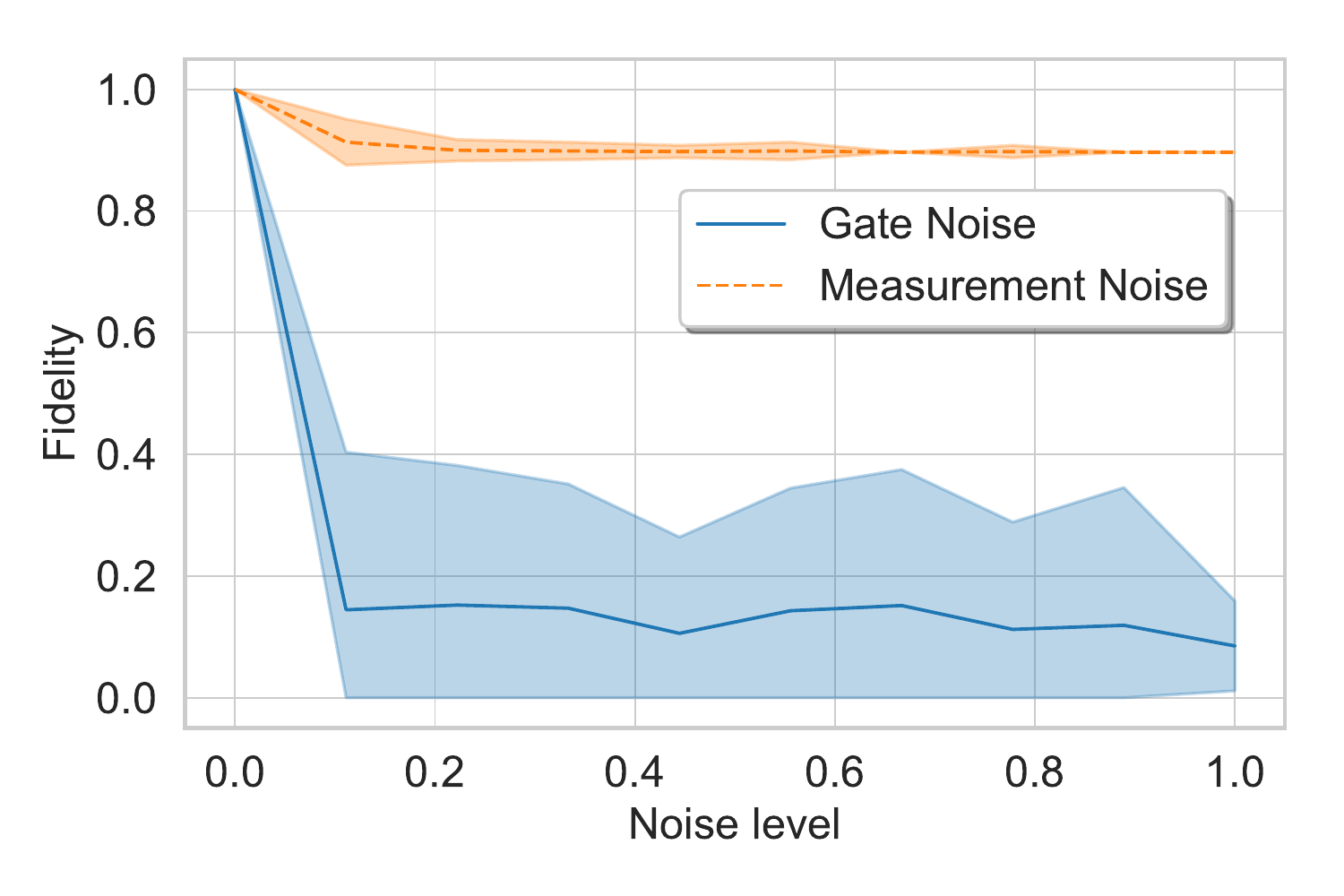}
    }
    \subfigure[]{
        \includegraphics[width=0.48\textwidth]{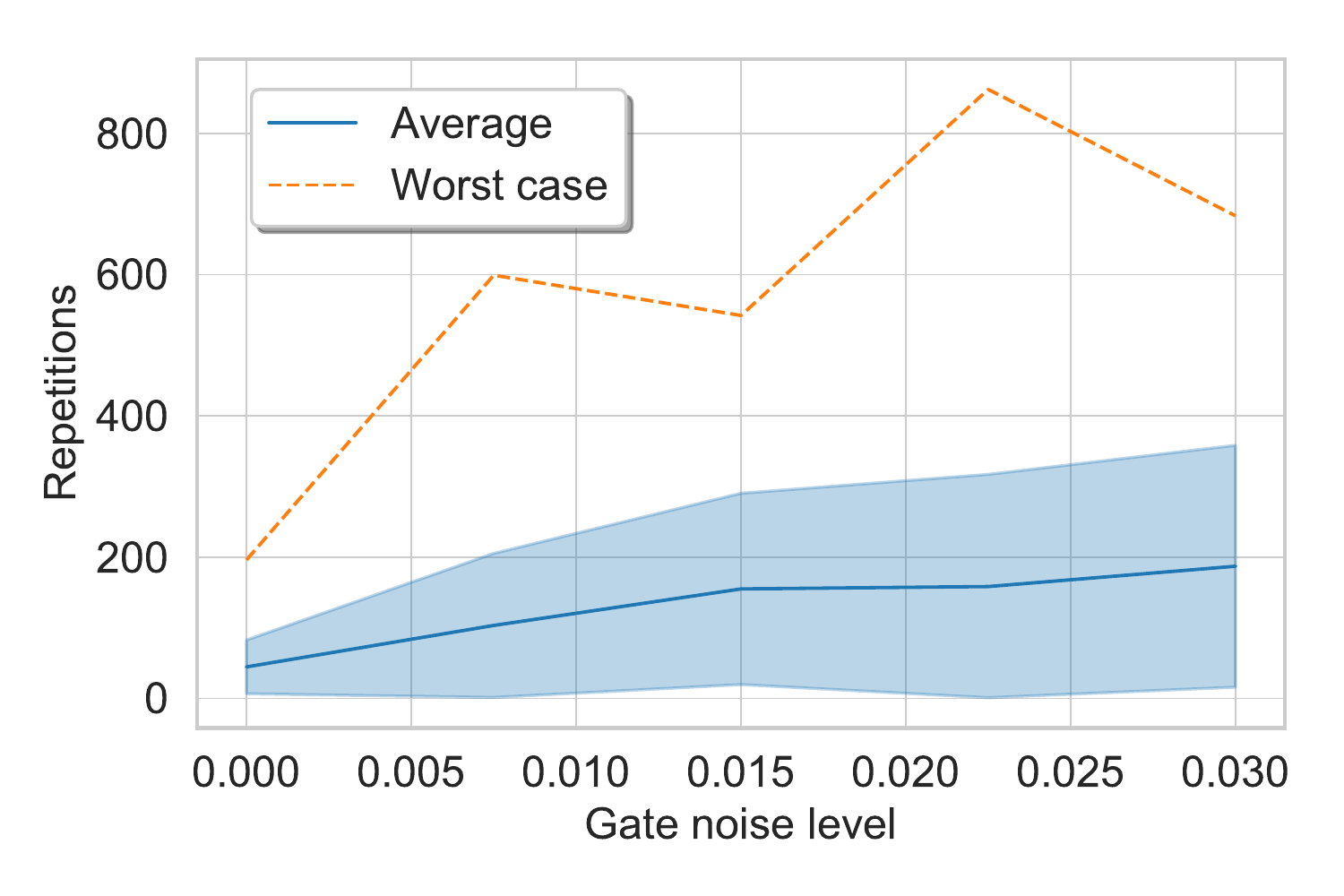}
    }
    \caption{Simulated gate and measurement noise on a specialized circuit for inverting the $2\times 2$ matrix $A$, run in Rigetti's quantum virtual machine. (a) The fidelity shows the overlap with the expected correct state after the computation. A fidelity of zero means that the output state (and hence the result of the computation) is completely orthogonal to the correct solution, while a fidelity of one means that the output state coincides with the expected one. (b) The number of repetitions expresses the average of how many times the probabilistic program is executed before it succeeds. We define a successful run with two conditions: the qubit in which the controlled rotation is performed is in the state $\ket{1}$ after its measurement, and the final state of the qubits has a fidelity greater than $0.9$ with the expected outcome of an ideal run of the protocol. The dashed line represents the maximum number of runs observed in the simulations before having a successful one.}
    \label{2by2simulated}
\end{figure*}

In this section, we report the results of the simulations conducted in Rigetti's quantum virtual machine.
We have conducted two sets of experiments to analyze the sensitivity of the protocol to different noise types that appear in real quantum computers.
In the first, we restrict ourselves to the simplest possible scenario of inverting the $2\times 2$ matrix $A=\frac{1}{2}\begin{pmatrix} 3 & 1 \\ 1 & 3 \end{pmatrix}$ with the problem-specific circuit in Ref.~\cite{cao2012quantum}. This circuit is much shallower than the full protocol detailed in Ref.~\cite{cao2013quantum}, making it more realistic to implement on current and near-future quantum computers due to its reduced depth.
The second case is the complete implementation of the full quantum matrix inversion protocol~\cite{Harrow2009a,cao2013quantum}.
This version requires a large number of ancilla qubits to perform the calculations, in particular for the computation of the reciprocals of the eigenvalues. We choose to simulate the inversion of a $4\times 4$ matrix with four bits of precision, which is the largest example that could fit on the largest Rigetti QPU.

We have studied the impact of two noise models, both being instances of parametric depolarizing noise.
The first one, known as \emph{gate noise}, applies a Pauli $X$ operator---which swaps the states $|0\rangle$ and $|1\rangle$ of the qubit it acts upon---with a certain probability on each qubit after \emph{every} gate application.
The probability of application of the operator indicates the noise level.
The second type of noise that we study is known as \emph{measurement noise}.
In this case, a Pauli $X$ operator is applied with certain probability only on every qubit that is measured, before the measurement takes place.
Therefore, it can also be understood as a readout error that, with a certain probability, instead of recording the result of a measurement, $y$, it records $NOT(y)$.

The circuits we implement have a much larger number of gates ($\sim$20 for the $2\times 2$ reduced version, increasing for the increasing size of the matrix being inverted) than measurements (just one, that which certifies the success of the eigenvalue inversion).
This is the reason why in all the experiments we run we observe that the gate noise has a stronger impact on the results than the measurement noise.

In Fig.~\ref{2by2simulated} we show the results for the inversion of the $2\times 2$ matrix $A$.
We analyze the two critical factors of the protocol, namely how different are the expected result of the protocol and the output from the simulator when we know that the inversion has succeeded, and how many times it is needed to run the protocol in order to obtain a successful run.
As expected, the measurement noise has a much smaller impact in the result than the gate noise, which for reasonably low noise levels already renders the output state (and hence the result of the inversion) with low overlap with the expected result.

\begin{figure*}[ht!]
    \centering
    \subfigure[]{
        \includegraphics[width=0.48\textwidth]{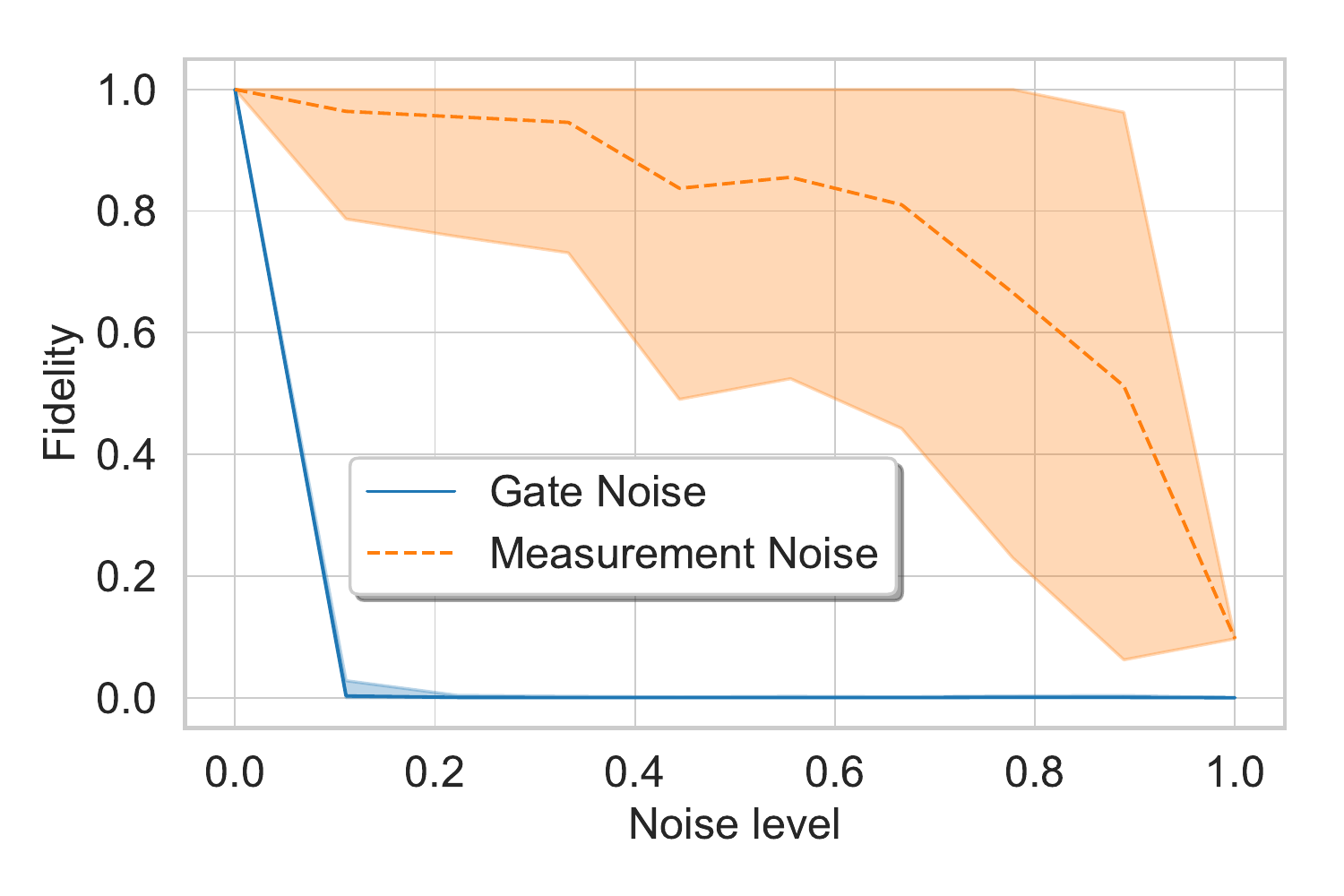}
    }
    \subfigure[]{
        \includegraphics[width=0.48\textwidth]{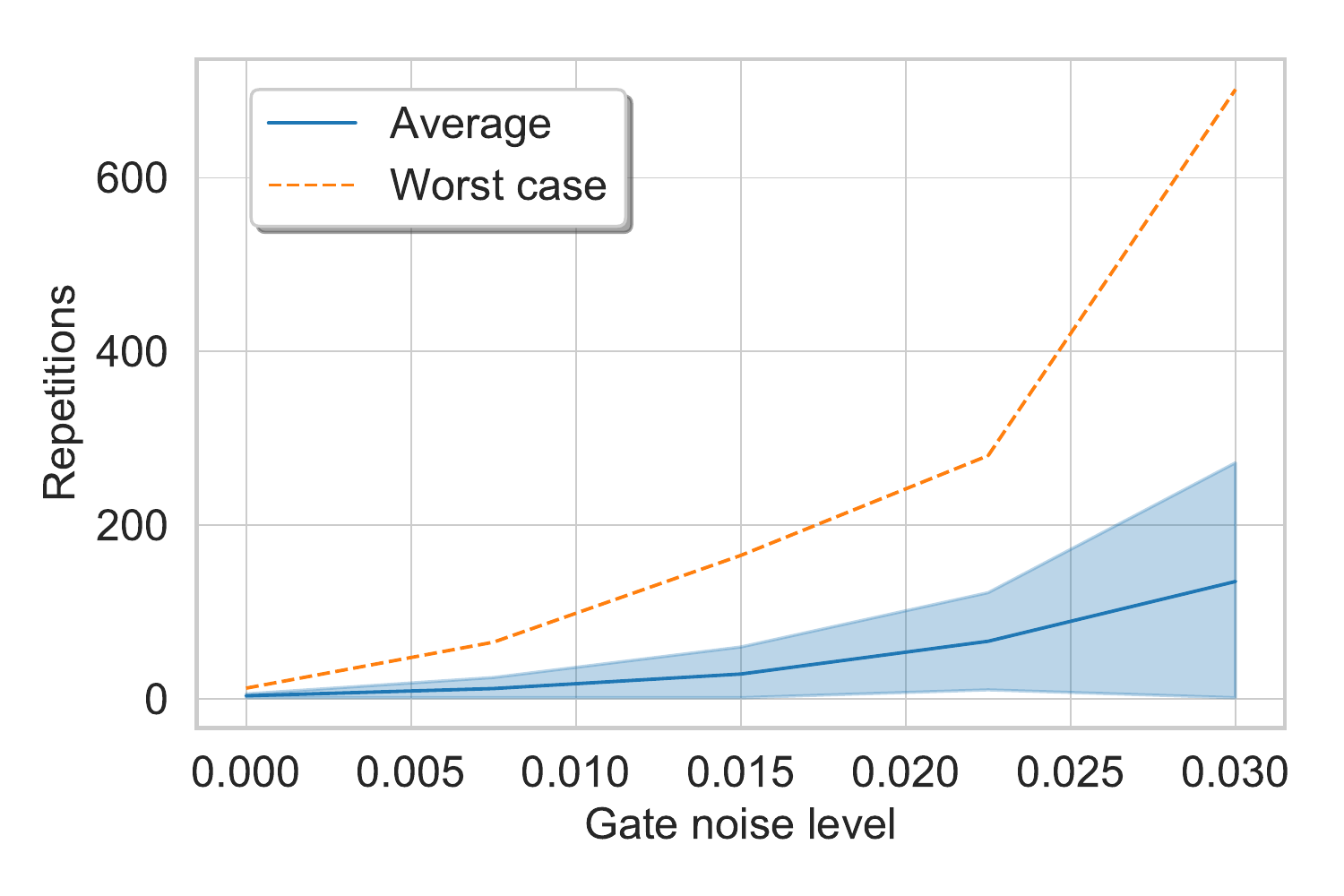}
    }
    \caption{Simulated gate and measurement noise on the generic circuit for inverting a matrix. The matrix in the benchmark was $4\times 4$, and the eigenvalues were represented by four bits of precision. Together with the ancilla qubits in the calculations, this is the largest system that can be simulated with less than 19 qubits, which is the size of Rigetti's largest QPU (19Q-Acorn).}
    \label{4by4simulated}
\end{figure*}

The number of repetitions needed for the algorithm to succeed, understood as the average number of times the algorithm must be run in order to obtain the outcome associated to the state $\ket{1}$ when measuring the qubit to which the conditional rotation is applied, is a fragile quantity that, on its own, does not provide meaningful insights when dealing with noise.
In the case of measurement noise, an error in the measurement either discards a successful run of the algorithm or accepts as successful a failed run, deeming further computations useless.
In the case of gate noise, even in the case the measurement succeeds and therefore the state of the flag qubit is $\ket{1}$, the remaining computations on the other qubits may lead to a final state that deviates from the expected result.

In order to obtain a good estimation of the number of runs needed to detect a final state that encodes the desired solution, in Fig.~\ref{2by2simulated}(b) we show the number of repetitions of the algorithm needed in order to have a successful run according to the flag qubit (i.e., that its state after the measurement is $\ket{1}$~\cite{Harrow2009a}), in which the overlap of the final state and the desired state is higher than a specific value.
We measure such an overlap with the fidelity, given by $\mathcal{F}=\left|\braket{\psi_\text{real}}{\psi_\text{ideal}}\right|^2$, where $\ket{\psi_\text{real}}$ and $\ket{\psi_\text{ideal}}$ determine the state of the qubits after a noisy simulation and a noiseless successful run, respectively.

Given that the protocol is probabilistic, the number of repetitions needed to have a successful run is dependent on the actual matrix to be inverted even in the case of a noiseless run, as can be observed by comparing Figs.~\ref{2by2simulated}(b) and~\ref{4by4simulated}(b), and grows fast with the gate noise level in the qubits.
It is important to track not only the average behavior of the protocol (in solid blue in Fig.~\ref{2by2simulated}(b)), but also worst-case scenarios (in dashed orange) where the protocol must be run up to more than five times the average in order to have a successful execution.
Nevertheless, worst-case performance scales with the noise level in a similar way as the average performance.

In Fig.~\ref{4by4simulated} we perform the same studies for the implementation of the general algorithm inverting a random $4\times 4$ matrix.
It is immediately apparent that increasing the circuit depth makes the protocol more sensitive to noise, and the fidelity drops to zero with lower variance in the case of the gate noise.
However, the noise level for which the fidelity of the output of the circuit with the expected state drops abruptly is approximately equal in both the $2\times 2$ and $4\times 4$ cases, and it would be interesting to see whether it remains constant for larger problems.
We still observe better robustness to measurement noise, but the impact of this kind of noise in the resulting state is stronger than in the problem-specific algorithm of Fig.~\ref{2by2simulated}.
The number of repetitions for a successful run now has a non-linear behavior with the level of gate noise in the simulation, although the ratio of the worst-case scenario to the average is the same as in the case of inverting the $2\times 2$ matrix.

\subsection{Evaluation on quantum processing units}
\label{actualqpu}
In this section, we implement the restricted $2\times 2$-matrix inversion algorithm in two real quantum computing architectures.
The reason of choosing the restricted algorithm is that current quantum computers have a small number of qubits, limited qubit-qubit connectivity, and most importantly, short coherence times, which implies that only shallow quantum circuits can be implemented.
The restricted algorithm be implemented with a much simpler circuit than the general one, resulting in about 20 gates for the full protocol~\cite{cao2012quantum}.

\begin{figure*}[ht!]
    \centering
        \includegraphics[width=0.8\textwidth]{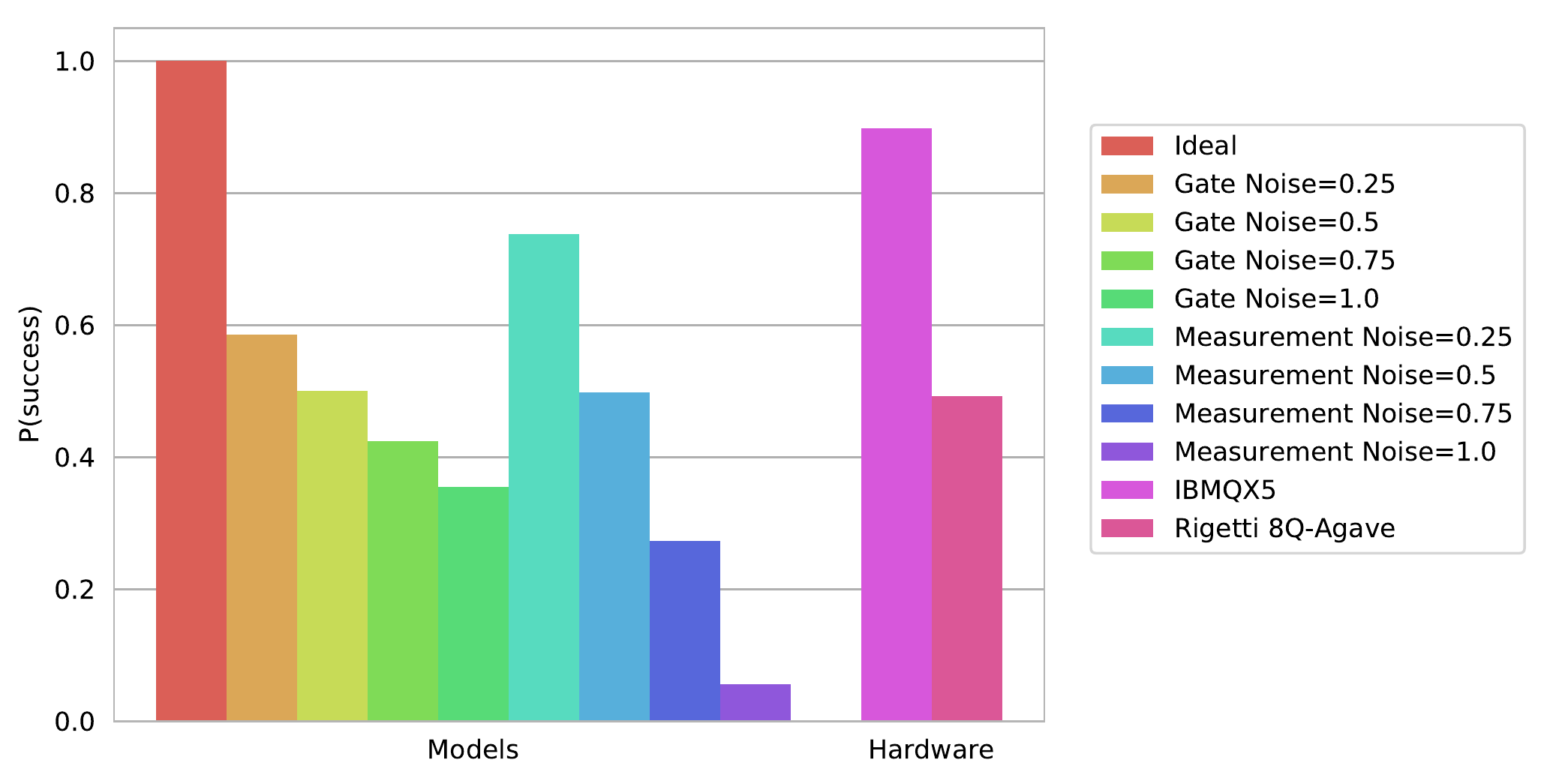}
    \caption{Probability of success of the swap test (i.e., the probability of the circuit result being the desired) after success in the eigenvalue inversion subroutine, for different classical noisy simulations, and executions on the IBM and Rigetti quantum processing units (rightmost bars). The noise models involve faulty gate operations---gate noise---and faulty readout errors---measurement noise---, with different probabilities of failure. The algorithm is run 8192 times for each instance, after which $P(success)$ is computed.}
    \label{qpus}
\end{figure*}

In the case of runs on real QPUs one does not have direct access to the whole output state of the circuit, but only to samples of measurements on it.
This makes it difficult to compute the fidelity with the expected state, and instead, we perform a \emph{swap test}~\cite{gottesman2001quantum}.
The test runs as follows: the expected result of the algorithm is encoded in auxiliary qubits, and after operations between the output and the expected result a flag qubit indicates whether both states are equal, in which case the state of the flag qubit is $|0\rangle$, or not, in which case the state is $|1\rangle$.
The figure of merit is now the probability of success in the test $P(success)=P(0)$, which can then be related to the fidelity by the expression $\mathcal{F}=|2P(success)-1|$.
Note that this success probability is different from the probability that the eigenvalue inversion subroutine succeeds, which is the quantity that has already been studied in Figs.~\ref{2by2simulated}(b) and~\ref{4by4simulated}(b).

We have implemented the protocol to be run in both Rigetti's 8Q-Agave and IBM's IBMQX5 quantum processing units.
The IBM QISKit software~\cite{cross2017open} also provides a classical simulator to run noisy experiments, and we use these to benchmark the performance of the runs on the real chips.
The results of the experiments can be found in Fig.~\ref{qpus}.

As in the case of the simulations in Rigetti's software stack, the measurement noise produces a smaller impact in the protocol than the gate noise.
Note that the qubit that encodes the success of the swap test is also subject to readout error when simulating measurement noise.
Therefore, for large measurement noise levels, the fact that $P(success) = P(0)$ is very low means that the actual state of the flag qubit is $|0\rangle$ (i.e., the protocol has succeeded, and the output state is the desired one), but due to the noise the result that is recorded after measuring is $1$.

Gate noise has a stronger impact in the final state.
This kind of error, unlike the measurement noise, does affect the computations in the circuit, so lower success probabilities now represent a real discrepancy between the output and desired states.
In this case, the success probabilities lie in the range of $[0.35, 0.6]$, which translates into fidelities in the range of $[0, 0.3]$.

Turning to executions in the real QPUs, the probability of protocol success is higher in IBMQX5.
This is mostly due to its improved coherence time\footnote{Information about performance measures of Rigetti's QPUs can be found in \href{http://docs.rigetti.com/en/1.9/qpu.html}{http://docs.rigetti.com/en/1.9/qpu.html}}\textsuperscript{,}\footnote{Information about performance measures of IBM's QPUs can be found in \href{http://www.research.ibm.com/ibm-q/technology/devices/}{http://www.research.ibm.com/ibm-q/technology/devices/}}, that allows keeping the state in the circuit better isolated from external perturbations during computation.
The probability of protocol success is $89\%$, which translates into a fidelity with the expected state of $0.78$. This is a very encouraging result, despite the size of the matrix inverted.
In contrast, the fidelity when the protocol is run in 8Q-Agave is close to zero, which means that all the information about the computation is lost during the process.
This is mainly due to the circuit depth being too large to maintain the quantum state isolated enough from the environment.

\section{Conclusions}
As quantum computers become available and continue improving in scale and noise tolerance, it is an exciting question to ask whether they can make a qualitative difference in machine learning applications.
Seminal works that explored this question focused on idealized, fully noise tolerant, large-scale quantum computers, and implemented simple machine learning algorithms like support vector machines and nearest-neighbor clustering.
However, an important fact is that for at least the next decade quantum computers will remain limited in scale and noise tolerance, and we must factor this in when we construct quantum-enhanced algorithms.
Furthermore, simple machine learning methods are already efficiently executed on classical hardware, so there is no need for the use of quantum algorithms in this case.

In this work, we studied a complex, Bayesian approach to deep architectures that is difficult to perform on digital hardware. We developed a quantum algorithm for learning Gaussian processes that can be applied layer by layer for training arbitrarily deep neural networks.
Furthermore, our protocol is a classical-quantum hybrid that largely removes the currently unrealistic technological requirements, such as a quantum random access memory.
The algorithm makes use of the quantum matrix inversion protocol which, albeit intricate, its mathematical assumptions are fulfilled by the kernel matrices originating from Gaussian processes.
In order to analyze the feasibility of a real use of the algorithm, we implemented its core routine, the quantum matrix inversion protocol, to be run in both quantum simulators and real state-of-the-art quantum processors. We observe that the accuracy of the protocol sharply drops with noise, but even with current, small quantum computers, high success rates can be achieved.

Although promising, these experimental results do not completely prove that the full protocol will be efficiently implementable in near-term quantum technologies. Full implementation in architectures with limited coherence time and sparse connectivity, as well as a fully-coherent variant of the training algorithm (which would have important applications in quantum simulations and quantum control), are interesting avenues for future research.

Not only are commercial quantum computers proliferating, but also the tools to program them, and, even more importantly, the collection of high-level algorithmic primitives~\cite{coles2018quantum}.
This enables machine learning researchers to leverage quantum technologies without the need of having an extensive background in quantum technologies.
Just as GPUs and efficient frameworks like TensorFlow~\cite{abadi2016tensorflow} and PyTorch~\cite{paszke2017pytorch} created an enormous community researching and deploying deep learning, we expect the same phenomenon will happen in the future with quantum processing units and collections of quantum algorithms.

\begin{acknowledgements}
We would like to thank Piotr Gawron (Polish Academy of Sciences), Will Zeng and Ryan Karle (Rigetti Computing), and Joseph Fitzsimons (SUTD and CQT) for discussions.
Z. Z. acknowledges support from Singapore's Ministry of Education and National Research Foundation under NRF Award NRF-NRFF2013-01.
The work of \mbox{A. P.-K.} is supported by Fundaci\'on Obra \mbox{Social} \mbox{``la Caixa''} (LCF/BQ/ES15/10360001), the Spanish MINECO (QIBEQI FIS2016-80773-P and Severo Ochoa SEV-2015-0522), Fundaci\'o Privada Cellex, and the Generalitat de Catalunya (SGR1381 and CERCA Program).
This research was supported by Perimeter Institute for Theoretical Physics. Research at Perimeter Institute is supported by the Government of Canada through Industry Canada and by the Province of Ontario through the Ministry of Economic Development and Innovation.
\end{acknowledgements}

\bibliography{bibliography}

\begin{thebibliography}{54}%
\makeatletter
\providecommand \@ifxundefined [1]{%
 \@ifx{#1\undefined}
}%
\providecommand \@ifnum [1]{%
 \ifnum #1\expandafter \@firstoftwo
 \else \expandafter \@secondoftwo
 \fi
}%
\providecommand \@ifx [1]{%
 \ifx #1\expandafter \@firstoftwo
 \else \expandafter \@secondoftwo
 \fi
}%
\providecommand \natexlab [1]{#1}%
\providecommand \enquote  [1]{``#1''}%
\providecommand \bibnamefont  [1]{#1}%
\providecommand \bibfnamefont [1]{#1}%
\providecommand \citenamefont [1]{#1}%
\providecommand \href@noop [0]{\@secondoftwo}%
\providecommand \href [0]{\begingroup \@sanitize@url \@href}%
\providecommand \@href[1]{\@@startlink{#1}\@@href}%
\providecommand \@@href[1]{\endgroup#1\@@endlink}%
\providecommand \@sanitize@url [0]{\catcode `\\12\catcode `\$12\catcode
  `\&12\catcode `\#12\catcode `\^12\catcode `\_12\catcode `\%12\relax}%
\providecommand \@@startlink[1]{}%
\providecommand \@@endlink[0]{}%
\providecommand \url  [0]{\begingroup\@sanitize@url \@url }%
\providecommand \@url [1]{\endgroup\@href {#1}{\urlprefix }}%
\providecommand \urlprefix  [0]{URL }%
\providecommand \Eprint [0]{\href }%
\providecommand \doibase [0]{http://dx.doi.org/}%
\providecommand \selectlanguage [0]{\@gobble}%
\providecommand \bibinfo  [0]{\@secondoftwo}%
\providecommand \bibfield  [0]{\@secondoftwo}%
\providecommand \translation [1]{[#1]}%
\providecommand \BibitemOpen [0]{}%
\providecommand \bibitemStop [0]{}%
\providecommand \bibitemNoStop [0]{.\EOS\space}%
\providecommand \EOS [0]{\spacefactor3000\relax}%
\providecommand \BibitemShut  [1]{\csname bibitem#1\endcsname}%
\let\auto@bib@innerbib\@empty
\bibitem [{\citenamefont {Ghahramani}(2015)}]{ghahramani2015probabilistic}%
  \BibitemOpen
  \bibfield  {author} {\bibinfo {author} {\bibfnamefont {Z.}~\bibnamefont
  {Ghahramani}},\ }\href {\doibase 10.1038/nature14541} {\bibfield  {journal}
  {\bibinfo  {journal} {Nature}\ }\textbf {\bibinfo {volume} {521}},\ \bibinfo
  {pages} {452–459} (\bibinfo {year} {2015})}\BibitemShut {NoStop}%
\bibitem [{\citenamefont {Bradshaw}\ \emph {et~al.}(2017)\citenamefont
  {Bradshaw}, \citenamefont {Matthews},\ and\ \citenamefont
  {Ghahramani}}]{bradshaw2017adversarial}%
  \BibitemOpen
  \bibfield  {author} {\bibinfo {author} {\bibfnamefont {J.}~\bibnamefont
  {Bradshaw}}, \bibinfo {author} {\bibfnamefont {A.~G. d.~G.}\ \bibnamefont
  {Matthews}}, \ and\ \bibinfo {author} {\bibfnamefont {Z.}~\bibnamefont
  {Ghahramani}},\ }\href@noop {} {\  (\bibinfo {year} {2017})},\ \Eprint
  {http://arxiv.org/abs/1707.02476} {arXiv:1707.02476} \BibitemShut {NoStop}%
\bibitem [{\citenamefont {Grosse}\ \emph {et~al.}(2017)\citenamefont {Grosse},
  \citenamefont {Pfaff}, \citenamefont {Smith},\ and\ \citenamefont
  {Backes}}]{grosse2017how}%
  \BibitemOpen
  \bibfield  {author} {\bibinfo {author} {\bibfnamefont {K.}~\bibnamefont
  {Grosse}}, \bibinfo {author} {\bibfnamefont {D.}~\bibnamefont {Pfaff}},
  \bibinfo {author} {\bibfnamefont {M.~T.}\ \bibnamefont {Smith}}, \ and\
  \bibinfo {author} {\bibfnamefont {M.}~\bibnamefont {Backes}},\ }\href@noop {}
  {\  (\bibinfo {year} {2017})},\ \Eprint {http://arxiv.org/abs/1711.06598}
  {arXiv:1711.06598} \BibitemShut {NoStop}%
\bibitem [{\citenamefont {Blundell}\ \emph {et~al.}(2015)\citenamefont
  {Blundell}, \citenamefont {Cornebise}, \citenamefont {Kavukcuoglu},\ and\
  \citenamefont {Wierstra}}]{blundell2015weight}%
  \BibitemOpen
  \bibfield  {author} {\bibinfo {author} {\bibfnamefont {C.}~\bibnamefont
  {Blundell}}, \bibinfo {author} {\bibfnamefont {J.}~\bibnamefont {Cornebise}},
  \bibinfo {author} {\bibfnamefont {K.}~\bibnamefont {Kavukcuoglu}}, \ and\
  \bibinfo {author} {\bibfnamefont {D.}~\bibnamefont {Wierstra}},\ }in\ \href
  {http://dl.acm.org/citation.cfm?id=3045118.3045290} {\emph {\bibinfo
  {booktitle} {Proceedings of the 32nd International Conference on
  International Conference on Machine Learning - Volume 37}}},\ \bibinfo
  {series and number} {ICML'15}\ (\bibinfo  {publisher} {JMLR.org},\ \bibinfo
  {year} {2015})\ pp.\ \bibinfo {pages} {1613--1622},\ \Eprint
  {http://arxiv.org/abs/1505.05424} {arXiv:1505.05424} \BibitemShut {NoStop}%
\bibitem [{\citenamefont {Gal}\ and\ \citenamefont
  {Ghahramani}(2016)}]{gal2016dropout}%
  \BibitemOpen
  \bibfield  {author} {\bibinfo {author} {\bibfnamefont {Y.}~\bibnamefont
  {Gal}}\ and\ \bibinfo {author} {\bibfnamefont {Z.}~\bibnamefont
  {Ghahramani}},\ }in\ \href {http://proceedings.mlr.press/v48/gal16.html}
  {\emph {\bibinfo {booktitle} {Proceedings of The 33rd International
  Conference on Machine Learning}}},\ \bibinfo {series} {Proceedings of Machine
  Learning Research}, Vol.~\bibinfo {volume} {48},\ \bibinfo {editor} {edited
  by\ \bibinfo {editor} {\bibfnamefont {M.~F.}\ \bibnamefont {Balcan}}\ and\
  \bibinfo {editor} {\bibfnamefont {K.~Q.}\ \bibnamefont {Weinberger}}}\
  (\bibinfo  {publisher} {PMLR},\ \bibinfo {address} {New York, New York,
  USA},\ \bibinfo {year} {2016})\ pp.\ \bibinfo {pages} {1050--1059},\ \Eprint
  {http://arxiv.org/abs/1506.02142} {arXiv:1506.02142} \BibitemShut {NoStop}%
\bibitem [{\citenamefont {Rasmussen}\ and\ \citenamefont
  {Williams}(2006)}]{rasmussen2006gaussian}%
  \BibitemOpen
  \bibfield  {author} {\bibinfo {author} {\bibfnamefont {C.~E.}\ \bibnamefont
  {Rasmussen}}\ and\ \bibinfo {author} {\bibfnamefont {C.~K.~I.}\ \bibnamefont
  {Williams}},\ }\href@noop {} {\emph {\bibinfo {title} {Gaussian Processes for
  Machine Learning}}}\ (\bibinfo  {publisher} {MIT Press},\ \bibinfo {year}
  {2006})\BibitemShut {NoStop}%
\bibitem [{\citenamefont {Lee}\ \emph {et~al.}(2018)\citenamefont {Lee},
  \citenamefont {Bahri}, \citenamefont {Novak}, \citenamefont {Schoenholz},
  \citenamefont {Pennington},\ and\ \citenamefont
  {{Sohl-Dickstein}}}]{lee2018deep}%
  \BibitemOpen
  \bibfield  {author} {\bibinfo {author} {\bibfnamefont {J.}~\bibnamefont
  {Lee}}, \bibinfo {author} {\bibfnamefont {Y.}~\bibnamefont {Bahri}}, \bibinfo
  {author} {\bibfnamefont {R.}~\bibnamefont {Novak}}, \bibinfo {author}
  {\bibfnamefont {S.~S.}\ \bibnamefont {Schoenholz}}, \bibinfo {author}
  {\bibfnamefont {J.}~\bibnamefont {Pennington}}, \ and\ \bibinfo {author}
  {\bibfnamefont {J.}~\bibnamefont {{Sohl-Dickstein}}},\ }in\ \href
  {https://openreview.net/forum?id=B1EA-M-0Z} {\emph {\bibinfo {booktitle}
  {International Conference on Learning Representations}}}\ (\bibinfo {year}
  {2018})\ \Eprint {http://arxiv.org/abs/1711.00165} {arXiv:1711.00165}
  \BibitemShut {NoStop}%
\bibitem [{\citenamefont {de~G.~Matthews}\ \emph {et~al.}(2018)\citenamefont
  {de~G.~Matthews}, \citenamefont {Hron}, \citenamefont {Rowland},
  \citenamefont {Turner},\ and\ \citenamefont
  {Ghahramani}}]{gmatthews2018gaussian}%
  \BibitemOpen
  \bibfield  {author} {\bibinfo {author} {\bibfnamefont {A.~G.}\ \bibnamefont
  {de~G.~Matthews}}, \bibinfo {author} {\bibfnamefont {J.}~\bibnamefont
  {Hron}}, \bibinfo {author} {\bibfnamefont {M.}~\bibnamefont {Rowland}},
  \bibinfo {author} {\bibfnamefont {R.~E.}\ \bibnamefont {Turner}}, \ and\
  \bibinfo {author} {\bibfnamefont {Z.}~\bibnamefont {Ghahramani}},\ }in\ \href
  {https://openreview.net/forum?id=H1-nGgWC-} {\emph {\bibinfo {booktitle}
  {International Conference on Learning Representations}}}\ (\bibinfo {year}
  {2018})\ \Eprint {http://arxiv.org/abs/1804.11271} {arXiv:1804.11271}
  \BibitemShut {NoStop}%
\bibitem [{\citenamefont {Verdon}\ \emph {et~al.}(2017)\citenamefont {Verdon},
  \citenamefont {Broughton},\ and\ \citenamefont
  {Biamonte}}]{verdon2017quantum}%
  \BibitemOpen
  \bibfield  {author} {\bibinfo {author} {\bibfnamefont {G.}~\bibnamefont
  {Verdon}}, \bibinfo {author} {\bibfnamefont {M.}~\bibnamefont {Broughton}}, \
  and\ \bibinfo {author} {\bibfnamefont {J.}~\bibnamefont {Biamonte}},\
  }\href@noop {} {\  (\bibinfo {year} {2017})},\ \Eprint
  {http://arxiv.org/abs/1712.05304} {arXiv:1712.05304} \BibitemShut {NoStop}%
\bibitem [{\citenamefont {Torrontegui}\ and\ \citenamefont
  {Garcia-Ripoll}(2018)}]{torrontegui2018universal}%
  \BibitemOpen
  \bibfield  {author} {\bibinfo {author} {\bibfnamefont {E.}~\bibnamefont
  {Torrontegui}}\ and\ \bibinfo {author} {\bibfnamefont {J.~J.}\ \bibnamefont
  {Garcia-Ripoll}},\ }\href@noop {} {\  (\bibinfo {year} {2018})},\ \Eprint
  {http://arxiv.org/abs/1801.00934} {arXiv:1801.00934} \BibitemShut {NoStop}%
\bibitem [{\citenamefont {Khoshaman}\ \emph {et~al.}(2018)\citenamefont
  {Khoshaman}, \citenamefont {Vinci}, \citenamefont {Denis}, \citenamefont
  {Andriyash},\ and\ \citenamefont {Amin}}]{khoshaman2018quantum}%
  \BibitemOpen
  \bibfield  {author} {\bibinfo {author} {\bibfnamefont {A.}~\bibnamefont
  {Khoshaman}}, \bibinfo {author} {\bibfnamefont {W.}~\bibnamefont {Vinci}},
  \bibinfo {author} {\bibfnamefont {B.}~\bibnamefont {Denis}}, \bibinfo
  {author} {\bibfnamefont {E.}~\bibnamefont {Andriyash}}, \ and\ \bibinfo
  {author} {\bibfnamefont {M.~H.}\ \bibnamefont {Amin}},\ }\href@noop {} {\
  (\bibinfo {year} {2018})},\ \Eprint {http://arxiv.org/abs/1802.05779}
  {arXiv:1802.05779} \BibitemShut {NoStop}%
\bibitem [{\citenamefont {Farhi}\ and\ \citenamefont
  {Neven}(2018)}]{farhi2018classification}%
  \BibitemOpen
  \bibfield  {author} {\bibinfo {author} {\bibfnamefont {E.}~\bibnamefont
  {Farhi}}\ and\ \bibinfo {author} {\bibfnamefont {H.}~\bibnamefont {Neven}},\
  }\href@noop {} {\  (\bibinfo {year} {2018})},\ \Eprint
  {http://arxiv.org/abs/1802.06002} {arXiv:1802.06002} \BibitemShut {NoStop}%
\bibitem [{\citenamefont {Verdon}\ \emph {et~al.}(2018)\citenamefont {Verdon},
  \citenamefont {Pye},\ and\ \citenamefont {Broughton}}]{verdon2018universal}%
  \BibitemOpen
  \bibfield  {author} {\bibinfo {author} {\bibfnamefont {G.}~\bibnamefont
  {Verdon}}, \bibinfo {author} {\bibfnamefont {J.}~\bibnamefont {Pye}}, \ and\
  \bibinfo {author} {\bibfnamefont {M.}~\bibnamefont {Broughton}},\ }\href@noop
  {} {\  (\bibinfo {year} {2018})},\ \Eprint {http://arxiv.org/abs/1806.09729}
  {arXiv:1806.09729} \BibitemShut {NoStop}%
\bibitem [{\citenamefont {Schuld}\ \emph {et~al.}(2014)\citenamefont {Schuld},
  \citenamefont {Sinayskiy},\ and\ \citenamefont
  {Petruccione}}]{schuld2014quest}%
  \BibitemOpen
  \bibfield  {author} {\bibinfo {author} {\bibfnamefont {M.}~\bibnamefont
  {Schuld}}, \bibinfo {author} {\bibfnamefont {I.}~\bibnamefont {Sinayskiy}}, \
  and\ \bibinfo {author} {\bibfnamefont {F.}~\bibnamefont {Petruccione}},\
  }\href {\doibase 10.1007/s11128-014-0809-8} {\bibfield  {journal} {\bibinfo
  {journal} {Quantum Inf. Process.}\ }\textbf {\bibinfo {volume} {13}},\
  \bibinfo {pages} {2567} (\bibinfo {year} {2014})},\ \Eprint
  {http://arxiv.org/abs/1408.7005} {arXiv:1408.7005} \BibitemShut {NoStop}%
\bibitem [{\citenamefont {Arjovsky}\ \emph {et~al.}(2015)\citenamefont
  {Arjovsky}, \citenamefont {Shah},\ and\ \citenamefont
  {Bengio}}]{arjovsky2015unitary}%
  \BibitemOpen
  \bibfield  {author} {\bibinfo {author} {\bibfnamefont {M.}~\bibnamefont
  {Arjovsky}}, \bibinfo {author} {\bibfnamefont {A.}~\bibnamefont {Shah}}, \
  and\ \bibinfo {author} {\bibfnamefont {Y.}~\bibnamefont {Bengio}},\
  }\href@noop {} {\  (\bibinfo {year} {2015})},\ \Eprint
  {http://arxiv.org/abs/1511.06464} {arXiv:1511.06464} \BibitemShut {NoStop}%
\bibitem [{\citenamefont {Hyland}\ and\ \citenamefont
  {R{ä}tsch}(2017)}]{hyland2016learning}%
  \BibitemOpen
  \bibfield  {author} {\bibinfo {author} {\bibfnamefont {S.}~\bibnamefont
  {Hyland}}\ and\ \bibinfo {author} {\bibfnamefont {G.}~\bibnamefont
  {R{ä}tsch}},\ }in\ \href
  {https://aaai.org/ocs/index.php/AAAI/AAAI17/paper/view/14930/14373} {\emph
  {\bibinfo {booktitle} {AAAI Conference on Artificial Intelligence}}}\
  (\bibinfo {year} {2017})\ \Eprint {http://arxiv.org/abs/1607.04903}
  {arXiv:1607.04903} \BibitemShut {NoStop}%
\bibitem [{\citenamefont {Liu}\ \emph {et~al.}(2017)\citenamefont {Liu},
  \citenamefont {Ran}, \citenamefont {Wittek}, \citenamefont {Peng},
  \citenamefont {Garc{\'\i}a}, \citenamefont {Su},\ and\ \citenamefont
  {Lewenstein}}]{liu2017machine}%
  \BibitemOpen
  \bibfield  {author} {\bibinfo {author} {\bibfnamefont {D.}~\bibnamefont
  {Liu}}, \bibinfo {author} {\bibfnamefont {S.-J.}\ \bibnamefont {Ran}},
  \bibinfo {author} {\bibfnamefont {P.}~\bibnamefont {Wittek}}, \bibinfo
  {author} {\bibfnamefont {C.}~\bibnamefont {Peng}}, \bibinfo {author}
  {\bibfnamefont {R.~B.}\ \bibnamefont {Garc{\'\i}a}}, \bibinfo {author}
  {\bibfnamefont {G.}~\bibnamefont {Su}}, \ and\ \bibinfo {author}
  {\bibfnamefont {M.}~\bibnamefont {Lewenstein}},\ }\href@noop {} {\  (\bibinfo
  {year} {2017})},\ \Eprint {http://arxiv.org/abs/1710.04833}
  {arXiv:1710.04833} \BibitemShut {NoStop}%
\bibitem [{\citenamefont {Stoudenmire}(2018)}]{stoudenmire2017learning}%
  \BibitemOpen
  \bibfield  {author} {\bibinfo {author} {\bibfnamefont {E.~M.}\ \bibnamefont
  {Stoudenmire}},\ }\href {\doibase 10.1088/2058-9565/aaba1a} {\bibfield
  {journal} {\bibinfo  {journal} {Quantum Sci. Technol.}\ }\textbf {\bibinfo
  {volume} {3}},\ \bibinfo {pages} {034003} (\bibinfo {year} {2018})},\ \Eprint
  {http://arxiv.org/abs/1801.00315} {arXiv:1801.00315} \BibitemShut {NoStop}%
\bibitem [{\citenamefont {Biamonte}\ \emph {et~al.}(2017)\citenamefont
  {Biamonte}, \citenamefont {Wittek}, \citenamefont {Pancotti}, \citenamefont
  {Rebentrost}, \citenamefont {Wiebe},\ and\ \citenamefont
  {Lloyd}}]{biamonte2017quantum}%
  \BibitemOpen
  \bibfield  {author} {\bibinfo {author} {\bibfnamefont {J.}~\bibnamefont
  {Biamonte}}, \bibinfo {author} {\bibfnamefont {P.}~\bibnamefont {Wittek}},
  \bibinfo {author} {\bibfnamefont {N.}~\bibnamefont {Pancotti}}, \bibinfo
  {author} {\bibfnamefont {P.}~\bibnamefont {Rebentrost}}, \bibinfo {author}
  {\bibfnamefont {N.}~\bibnamefont {Wiebe}}, \ and\ \bibinfo {author}
  {\bibfnamefont {S.}~\bibnamefont {Lloyd}},\ }\href {\doibase
  10.1038/nature23474} {\bibfield  {journal} {\bibinfo  {journal} {Nature}\
  }\textbf {\bibinfo {volume} {549}},\ \bibinfo {pages} {195–202} (\bibinfo
  {year} {2017})},\ \Eprint {http://arxiv.org/abs/1611.09347}
  {arXiv:1611.09347} \BibitemShut {NoStop}%
\bibitem [{\citenamefont {Trabelsi}\ \emph {et~al.}(2017)\citenamefont
  {Trabelsi}, \citenamefont {Bilaniuk}, \citenamefont {Zhang}, \citenamefont
  {Serdyuk}, \citenamefont {Subramanian}, \citenamefont {Santos}, \citenamefont
  {Mehri}, \citenamefont {Rostamzadeh}, \citenamefont {Bengio},\ and\
  \citenamefont {Pal}}]{trabelsi2017deep}%
  \BibitemOpen
  \bibfield  {author} {\bibinfo {author} {\bibfnamefont {C.}~\bibnamefont
  {Trabelsi}}, \bibinfo {author} {\bibfnamefont {O.}~\bibnamefont {Bilaniuk}},
  \bibinfo {author} {\bibfnamefont {Y.}~\bibnamefont {Zhang}}, \bibinfo
  {author} {\bibfnamefont {D.}~\bibnamefont {Serdyuk}}, \bibinfo {author}
  {\bibfnamefont {S.}~\bibnamefont {Subramanian}}, \bibinfo {author}
  {\bibfnamefont {J.~F.}\ \bibnamefont {Santos}}, \bibinfo {author}
  {\bibfnamefont {S.}~\bibnamefont {Mehri}}, \bibinfo {author} {\bibfnamefont
  {N.}~\bibnamefont {Rostamzadeh}}, \bibinfo {author} {\bibfnamefont
  {Y.}~\bibnamefont {Bengio}}, \ and\ \bibinfo {author} {\bibfnamefont {C.~J.}\
  \bibnamefont {Pal}},\ }\href@noop {} {\  (\bibinfo {year} {2017})},\ \Eprint
  {http://arxiv.org/abs/1705.09792} {arXiv:1705.09792} \BibitemShut {NoStop}%
\bibitem [{\citenamefont {Zhao}\ \emph {et~al.}(2015)\citenamefont {Zhao},
  \citenamefont {Fitzsimons},\ and\ \citenamefont
  {Fitzsimons}}]{zhao2015quantum}%
  \BibitemOpen
  \bibfield  {author} {\bibinfo {author} {\bibfnamefont {Z.}~\bibnamefont
  {Zhao}}, \bibinfo {author} {\bibfnamefont {J.~K.}\ \bibnamefont
  {Fitzsimons}}, \ and\ \bibinfo {author} {\bibfnamefont {J.~F.}\ \bibnamefont
  {Fitzsimons}},\ }\href@noop {} {\  (\bibinfo {year} {2015})},\ \Eprint
  {http://arxiv.org/abs/1512.03929} {arXiv:1512.03929} \BibitemShut {NoStop}%
\bibitem [{\citenamefont {Smith}\ \emph {et~al.}(2016)\citenamefont {Smith},
  \citenamefont {Curtis},\ and\ \citenamefont {Zeng}}]{smith2016practical}%
  \BibitemOpen
  \bibfield  {author} {\bibinfo {author} {\bibfnamefont {R.~S.}\ \bibnamefont
  {Smith}}, \bibinfo {author} {\bibfnamefont {M.~J.}\ \bibnamefont {Curtis}}, \
  and\ \bibinfo {author} {\bibfnamefont {W.~J.}\ \bibnamefont {Zeng}},\
  }\href@noop {} {\  (\bibinfo {year} {2016})},\ \Eprint
  {http://arxiv.org/abs/1608.03355} {arXiv:1608.03355} \BibitemShut {NoStop}%
\bibitem [{\citenamefont {Cross}\ \emph {et~al.}(2017)\citenamefont {Cross},
  \citenamefont {Bishop}, \citenamefont {Smolin},\ and\ \citenamefont
  {Gambetta}}]{cross2017open}%
  \BibitemOpen
  \bibfield  {author} {\bibinfo {author} {\bibfnamefont {A.~W.}\ \bibnamefont
  {Cross}}, \bibinfo {author} {\bibfnamefont {L.~S.}\ \bibnamefont {Bishop}},
  \bibinfo {author} {\bibfnamefont {J.~A.}\ \bibnamefont {Smolin}}, \ and\
  \bibinfo {author} {\bibfnamefont {J.~M.}\ \bibnamefont {Gambetta}},\
  }\href@noop {} {\  (\bibinfo {year} {2017})},\ \Eprint
  {http://arxiv.org/abs/1707.03429} {arXiv:1707.03429} \BibitemShut {NoStop}%
\bibitem [{\citenamefont {Neal}(1994)}]{neal1994priors}%
  \BibitemOpen
  \bibfield  {author} {\bibinfo {author} {\bibfnamefont {R.~M.}\ \bibnamefont
  {Neal}},\ }\href@noop {} {\emph {\bibinfo {title} {Priors for infinite
  networks}}},\ \bibinfo {type} {Tech. Rep.}\ \bibinfo {number} {crg-tr-94-1}\
  (\bibinfo  {institution} {University of Toronto},\ \bibinfo {year}
  {1994})\BibitemShut {NoStop}%
\bibitem [{\citenamefont {Harrow}\ \emph {et~al.}(2009)\citenamefont {Harrow},
  \citenamefont {Hassidim},\ and\ \citenamefont {Lloyd}}]{Harrow2009a}%
  \BibitemOpen
  \bibfield  {author} {\bibinfo {author} {\bibfnamefont {A.~W.}\ \bibnamefont
  {Harrow}}, \bibinfo {author} {\bibfnamefont {A.}~\bibnamefont {Hassidim}}, \
  and\ \bibinfo {author} {\bibfnamefont {S.}~\bibnamefont {Lloyd}},\ }\href
  {\doibase 10.1103/PhysRevLett.103.150502} {\bibfield  {journal} {\bibinfo
  {journal} {Phys. Rev. Lett.}\ }\textbf {\bibinfo {volume} {103}},\ \bibinfo
  {pages} {150502} (\bibinfo {year} {2009})},\ \Eprint
  {http://arxiv.org/abs/0811.3171} {arXiv:0811.3171} \BibitemShut {NoStop}%
\bibitem [{\citenamefont {Kitaev}(1995)}]{Kitaev1995}%
  \BibitemOpen
  \bibfield  {author} {\bibinfo {author} {\bibfnamefont {A.~Y.}\ \bibnamefont
  {Kitaev}},\ }\href@noop {} {\  (\bibinfo {year} {1995})},\ \Eprint
  {http://arxiv.org/abs/quant-ph/9511026} {arXiv:quant-ph/9511026} \BibitemShut
  {NoStop}%
\bibitem [{\citenamefont {Tacchino}\ \emph {et~al.}(2018)\citenamefont
  {Tacchino}, \citenamefont {Macchiavello}, \citenamefont {Gerace},\ and\
  \citenamefont {Bajoni}}]{tacchino2018}%
  \BibitemOpen
  \bibfield  {author} {\bibinfo {author} {\bibfnamefont {F.}~\bibnamefont
  {Tacchino}}, \bibinfo {author} {\bibfnamefont {C.}~\bibnamefont
  {Macchiavello}}, \bibinfo {author} {\bibfnamefont {D.}~\bibnamefont
  {Gerace}}, \ and\ \bibinfo {author} {\bibfnamefont {D.}~\bibnamefont
  {Bajoni}},\ }\href@noop {} {\  (\bibinfo {year} {2018})},\ \Eprint
  {http://arxiv.org/abs/1811.02266v1} {arXiv:1811.02266v1} \BibitemShut
  {NoStop}%
\bibitem [{\citenamefont {Schuld}\ and\ \citenamefont
  {Killoran}(2018)}]{schuld2018FH}%
  \BibitemOpen
  \bibfield  {author} {\bibinfo {author} {\bibfnamefont {M.}~\bibnamefont
  {Schuld}}\ and\ \bibinfo {author} {\bibfnamefont {N.}~\bibnamefont
  {Killoran}},\ }\href {\doibase 10.1103/PhysRevLett.122.040504} {\bibfield
  {journal} {\bibinfo  {journal} {Phys. Rev. Lett.}\ }\textbf {\bibinfo
  {volume} {122}},\ \bibinfo {pages} {040504} (\bibinfo {year} {2018})},\
  \Eprint {http://arxiv.org/abs/1803.07128} {arXiv:1803.07128} \BibitemShut
  {NoStop}%
\bibitem [{\citenamefont {Lloyd}(1996)}]{lloyd1996universal}%
  \BibitemOpen
  \bibfield  {author} {\bibinfo {author} {\bibfnamefont {S.}~\bibnamefont
  {Lloyd}},\ }\href {\doibase 10.1126/science.273.5278.1073} {\bibfield
  {journal} {\bibinfo  {journal} {Science}\ }\textbf {\bibinfo {volume}
  {273}},\ \bibinfo {pages} {1073–1078} (\bibinfo {year} {1996})}\BibitemShut
  {NoStop}%
\bibitem [{\citenamefont {Childs}(2010)}]{childs2010relationship}%
  \BibitemOpen
  \bibfield  {author} {\bibinfo {author} {\bibfnamefont {A.~M.}\ \bibnamefont
  {Childs}},\ }\href {\doibase 10.1007/s00220-009-0930-1} {\bibfield  {journal}
  {\bibinfo  {journal} {Commun. Math. Phys.}\ }\textbf {\bibinfo {volume}
  {294}},\ \bibinfo {pages} {581} (\bibinfo {year} {2010})},\ \Eprint
  {http://arxiv.org/abs/0810.0312} {arXiv:0810.0312} \BibitemShut {NoStop}%
\bibitem [{\citenamefont {Berry}\ and\ \citenamefont
  {Childs}(2012)}]{berry2012black}%
  \BibitemOpen
  \bibfield  {author} {\bibinfo {author} {\bibfnamefont {D.~W.}\ \bibnamefont
  {Berry}}\ and\ \bibinfo {author} {\bibfnamefont {A.~M.}\ \bibnamefont
  {Childs}},\ }\href {\doibase 10.26421/QIC12.1-2} {\bibfield  {journal}
  {\bibinfo  {journal} {Quantum Inf. Comput.}\ }\textbf {\bibinfo {volume}
  {12}},\ \bibinfo {pages} {29} (\bibinfo {year} {2012})},\ \Eprint
  {http://arxiv.org/abs/0910.4157} {arXiv:0910.4157} \BibitemShut {NoStop}%
\bibitem [{\citenamefont {Furrer}\ \emph {et~al.}(2006)\citenamefont {Furrer},
  \citenamefont {Genton},\ and\ \citenamefont {Nychka}}]{furrer2006covariance}%
  \BibitemOpen
  \bibfield  {author} {\bibinfo {author} {\bibfnamefont {R.}~\bibnamefont
  {Furrer}}, \bibinfo {author} {\bibfnamefont {M.~G.}\ \bibnamefont {Genton}},
  \ and\ \bibinfo {author} {\bibfnamefont {D.}~\bibnamefont {Nychka}},\ }\href
  {\doibase 10.1198/106186006x132178} {\bibfield  {journal} {\bibinfo
  {journal} {J. Comput. Graph. Stat.}\ }\textbf {\bibinfo {volume} {15}},\
  \bibinfo {pages} {502} (\bibinfo {year} {2006})}\BibitemShut {NoStop}%
\bibitem [{\citenamefont {Wittek}\ and\ \citenamefont
  {Tan}(2011)}]{wittek2011compact}%
  \BibitemOpen
  \bibfield  {author} {\bibinfo {author} {\bibfnamefont {P.}~\bibnamefont
  {Wittek}}\ and\ \bibinfo {author} {\bibfnamefont {C.~L.}\ \bibnamefont
  {Tan}},\ }\href {\doibase 10.1109/TPAMI.2011.28} {\bibfield  {journal}
  {\bibinfo  {journal} {Trans. Pattern Anal. Mach. Intell.}\ }\textbf {\bibinfo
  {volume} {33}},\ \bibinfo {pages} {2039–2050} (\bibinfo {year}
  {2011})}\BibitemShut {NoStop}%
\bibitem [{\citenamefont {Wossnig}\ \emph {et~al.}(2018)\citenamefont
  {Wossnig}, \citenamefont {Zhao},\ and\ \citenamefont
  {Prakash}}]{wossnig2017quantum}%
  \BibitemOpen
  \bibfield  {author} {\bibinfo {author} {\bibfnamefont {L.}~\bibnamefont
  {Wossnig}}, \bibinfo {author} {\bibfnamefont {Z.}~\bibnamefont {Zhao}}, \
  and\ \bibinfo {author} {\bibfnamefont {A.}~\bibnamefont {Prakash}},\ }\href
  {\doibase 10.1103/PhysRevLett.120.050502} {\bibfield  {journal} {\bibinfo
  {journal} {Phys. Rev. Lett.}\ }\textbf {\bibinfo {volume} {120}},\ \bibinfo
  {pages} {050502} (\bibinfo {year} {2018})},\ \Eprint
  {http://arxiv.org/abs/1704.06174} {arXiv:1704.06174} \BibitemShut {NoStop}%
\bibitem [{\citenamefont {Zhao}\ \emph {et~al.}(2018)\citenamefont {Zhao},
  \citenamefont {Fitzsimons}, \citenamefont {Osborne}, \citenamefont
  {Roberts},\ and\ \citenamefont {Fitzsimons}}]{zhao2018quantum}%
  \BibitemOpen
  \bibfield  {author} {\bibinfo {author} {\bibfnamefont {Z.}~\bibnamefont
  {Zhao}}, \bibinfo {author} {\bibfnamefont {J.~K.}\ \bibnamefont
  {Fitzsimons}}, \bibinfo {author} {\bibfnamefont {M.~A.}\ \bibnamefont
  {Osborne}}, \bibinfo {author} {\bibfnamefont {S.~J.}\ \bibnamefont
  {Roberts}}, \ and\ \bibinfo {author} {\bibfnamefont {J.~F.}\ \bibnamefont
  {Fitzsimons}},\ }\href@noop {} {\  (\bibinfo {year} {2018})},\ \Eprint
  {http://arxiv.org/abs/1803.10520} {arXiv:1803.10520} \BibitemShut {NoStop}%
\bibitem [{\citenamefont {Cho}\ and\ \citenamefont
  {Saul}(2009)}]{cho2009kernel}%
  \BibitemOpen
  \bibfield  {author} {\bibinfo {author} {\bibfnamefont {Y.}~\bibnamefont
  {Cho}}\ and\ \bibinfo {author} {\bibfnamefont {L.~K.}\ \bibnamefont {Saul}},\
  }in\ \href
  {https://papers.nips.cc/paper/3628-kernel-methods-for-deep-learning.pdf}
  {\emph {\bibinfo {booktitle} {Advances in Neural Information Processing
  Systems}}}\ (\bibinfo {year} {2009})\ pp.\ \bibinfo {pages}
  {342--350}\BibitemShut {NoStop}%
\bibitem [{\citenamefont {Daniely}\ \emph {et~al.}(2016)\citenamefont
  {Daniely}, \citenamefont {Frostig},\ and\ \citenamefont
  {Singer}}]{daniely2016toward}%
  \BibitemOpen
  \bibfield  {author} {\bibinfo {author} {\bibfnamefont {A.}~\bibnamefont
  {Daniely}}, \bibinfo {author} {\bibfnamefont {R.}~\bibnamefont {Frostig}}, \
  and\ \bibinfo {author} {\bibfnamefont {Y.}~\bibnamefont {Singer}},\
  }\href@noop {} {\  (\bibinfo {year} {2016})},\ \Eprint
  {http://arxiv.org/abs/1602.05897} {arXiv:1602.05897} \BibitemShut {NoStop}%
\bibitem [{\citenamefont {Glorot}\ \emph {et~al.}(2011)\citenamefont {Glorot},
  \citenamefont {Bordes},\ and\ \citenamefont {Bengio}}]{glorot2011relu}%
  \BibitemOpen
  \bibfield  {author} {\bibinfo {author} {\bibfnamefont {X.}~\bibnamefont
  {Glorot}}, \bibinfo {author} {\bibfnamefont {A.}~\bibnamefont {Bordes}}, \
  and\ \bibinfo {author} {\bibfnamefont {Y.}~\bibnamefont {Bengio}},\ }in\
  \href {http://proceedings.mlr.press/v15/glorot11a.html} {\emph {\bibinfo
  {booktitle} {Proceedings of the Fourteenth International Conference on
  Artificial Intelligence and Statistics}}},\ \bibinfo {series} {Proceedings of
  Machine Learning Research}, Vol.~\bibinfo {volume} {15},\ \bibinfo {editor}
  {edited by\ \bibinfo {editor} {\bibfnamefont {G.}~\bibnamefont {Gordon}},
  \bibinfo {editor} {\bibfnamefont {D.}~\bibnamefont {Dunson}}, \ and\ \bibinfo
  {editor} {\bibfnamefont {M.}~\bibnamefont {Dudík}}}\ (\bibinfo  {publisher}
  {PMLR},\ \bibinfo {address} {Fort Lauderdale, FL, USA},\ \bibinfo {year}
  {2011})\ pp.\ \bibinfo {pages} {315--323}\BibitemShut {NoStop}%
\bibitem [{\citenamefont {Rebentrost}\ \emph {et~al.}(2014)\citenamefont
  {Rebentrost}, \citenamefont {Mohseni},\ and\ \citenamefont
  {Lloyd}}]{rebentrost2014quantum}%
  \BibitemOpen
  \bibfield  {author} {\bibinfo {author} {\bibfnamefont {P.}~\bibnamefont
  {Rebentrost}}, \bibinfo {author} {\bibfnamefont {M.}~\bibnamefont {Mohseni}},
  \ and\ \bibinfo {author} {\bibfnamefont {S.}~\bibnamefont {Lloyd}},\ }\href
  {\doibase 10.1103/PhysRevLett.113.130503} {\bibfield  {journal} {\bibinfo
  {journal} {Phys. Rev. Lett.}\ }\textbf {\bibinfo {volume} {113}},\ \bibinfo
  {pages} {130503} (\bibinfo {year} {2014})},\ \Eprint
  {http://arxiv.org/abs/1307.0471} {arXiv:1307.0471} \BibitemShut {NoStop}%
\bibitem [{\citenamefont {Nielsen}\ and\ \citenamefont
  {Chuang}(2000)}]{nielsen2000quantum}%
  \BibitemOpen
  \bibfield  {author} {\bibinfo {author} {\bibfnamefont {M.~A.}\ \bibnamefont
  {Nielsen}}\ and\ \bibinfo {author} {\bibfnamefont {I.~L.}\ \bibnamefont
  {Chuang}},\ }\href@noop {} {\emph {\bibinfo {title} {Quantum computation and
  quantum information}}}\ (\bibinfo  {publisher} {Cambridge University Press},\
  \bibinfo {year} {2000})\BibitemShut {NoStop}%
\bibitem [{\citenamefont {Berry}\ \emph {et~al.}(2015)\citenamefont {Berry},
  \citenamefont {Childs},\ and\ \citenamefont
  {Kothari}}]{berry2015hamiltonian}%
  \BibitemOpen
  \bibfield  {author} {\bibinfo {author} {\bibfnamefont {D.~W.}\ \bibnamefont
  {Berry}}, \bibinfo {author} {\bibfnamefont {A.~M.}\ \bibnamefont {Childs}}, \
  and\ \bibinfo {author} {\bibfnamefont {R.}~\bibnamefont {Kothari}},\ }in\
  \href {\doibase 10.1109/FOCS.2015.54} {\emph {\bibinfo {booktitle}
  {Proceedings of FOCS-15, 56th Annual Symposium on Foundations of Computer
  Science}}}\ (\bibinfo {year} {2015})\ pp.\ \bibinfo {pages} {792--809},\
  \Eprint {http://arxiv.org/abs/1501.01715} {arXiv:1501.01715} \BibitemShut
  {NoStop}%
\bibitem [{\citenamefont {Kimmel}\ \emph {et~al.}(2017)\citenamefont {Kimmel},
  \citenamefont {Lin}, \citenamefont {Low}, \citenamefont {Ozols},\ and\
  \citenamefont {Yoder}}]{kimmel2017hamiltonian}%
  \BibitemOpen
  \bibfield  {author} {\bibinfo {author} {\bibfnamefont {S.}~\bibnamefont
  {Kimmel}}, \bibinfo {author} {\bibfnamefont {C.~Y.-Y.}\ \bibnamefont {Lin}},
  \bibinfo {author} {\bibfnamefont {G.~H.}\ \bibnamefont {Low}}, \bibinfo
  {author} {\bibfnamefont {M.}~\bibnamefont {Ozols}}, \ and\ \bibinfo {author}
  {\bibfnamefont {T.~J.}\ \bibnamefont {Yoder}},\ }\href {\doibase
  10.1038/s41534-017-0013-7} {\bibfield  {journal} {\bibinfo  {journal} {npj
  Quantum Inf.}\ }\textbf {\bibinfo {volume} {3}},\ \bibinfo {pages} {13}
  (\bibinfo {year} {2017})},\ \Eprint {http://arxiv.org/abs/1608.00281}
  {arXiv:1608.00281} \BibitemShut {NoStop}%
\bibitem [{\citenamefont {Rebentrost}\ \emph {et~al.}(2018)\citenamefont
  {Rebentrost}, \citenamefont {Steffens}, \citenamefont {Marvian},\ and\
  \citenamefont {Lloyd}}]{rebentrost2018quantum}%
  \BibitemOpen
  \bibfield  {author} {\bibinfo {author} {\bibfnamefont {P.}~\bibnamefont
  {Rebentrost}}, \bibinfo {author} {\bibfnamefont {A.}~\bibnamefont
  {Steffens}}, \bibinfo {author} {\bibfnamefont {I.}~\bibnamefont {Marvian}}, \
  and\ \bibinfo {author} {\bibfnamefont {S.}~\bibnamefont {Lloyd}},\ }\href
  {\doibase 10.1103/PhysRevA.97.012327} {\bibfield  {journal} {\bibinfo
  {journal} {Phys. Rev. A}\ }\textbf {\bibinfo {volume} {97}},\ \bibinfo
  {pages} {012327} (\bibinfo {year} {2018})},\ \Eprint
  {http://arxiv.org/abs/1607.05404} {arXiv:1607.05404} \BibitemShut {NoStop}%
\bibitem [{\citenamefont {Rebentrost}\ \emph {et~al.}(2016)\citenamefont
  {Rebentrost}, \citenamefont {Schuld}, \citenamefont {Wossnig}, \citenamefont
  {Petruccione},\ and\ \citenamefont {Lloyd}}]{rebentrost2016quantum}%
  \BibitemOpen
  \bibfield  {author} {\bibinfo {author} {\bibfnamefont {P.}~\bibnamefont
  {Rebentrost}}, \bibinfo {author} {\bibfnamefont {M.}~\bibnamefont {Schuld}},
  \bibinfo {author} {\bibfnamefont {L.}~\bibnamefont {Wossnig}}, \bibinfo
  {author} {\bibfnamefont {F.}~\bibnamefont {Petruccione}}, \ and\ \bibinfo
  {author} {\bibfnamefont {S.}~\bibnamefont {Lloyd}},\ }\href@noop {} {\
  (\bibinfo {year} {2016})},\ \Eprint {http://arxiv.org/abs/1612.01789}
  {arXiv:1612.01789} \BibitemShut {NoStop}%
\bibitem [{\citenamefont {Suzuki}(1992)}]{suzuki1992general}%
  \BibitemOpen
  \bibfield  {author} {\bibinfo {author} {\bibfnamefont {M.}~\bibnamefont
  {Suzuki}},\ }\href {\doibase 10.1016/0375-9601(92)90335-J} {\bibfield
  {journal} {\bibinfo  {journal} {Phys. Lett. A}\ }\textbf {\bibinfo {volume}
  {165}},\ \bibinfo {pages} {387} (\bibinfo {year} {1992})}\BibitemShut
  {NoStop}%
\bibitem [{\citenamefont {Childs}\ \emph {et~al.}(2003)\citenamefont {Childs},
  \citenamefont {Cleve}, \citenamefont {Deotto}, \citenamefont {Farhi},
  \citenamefont {Gutmann},\ and\ \citenamefont
  {Spielman}}]{childs2003exponential}%
  \BibitemOpen
  \bibfield  {author} {\bibinfo {author} {\bibfnamefont {A.~M.}\ \bibnamefont
  {Childs}}, \bibinfo {author} {\bibfnamefont {R.}~\bibnamefont {Cleve}},
  \bibinfo {author} {\bibfnamefont {E.}~\bibnamefont {Deotto}}, \bibinfo
  {author} {\bibfnamefont {E.}~\bibnamefont {Farhi}}, \bibinfo {author}
  {\bibfnamefont {S.}~\bibnamefont {Gutmann}}, \ and\ \bibinfo {author}
  {\bibfnamefont {D.~A.}\ \bibnamefont {Spielman}},\ }in\ \href {\doibase
  10.1145/780542.780552} {\emph {\bibinfo {booktitle} {Proceedings of STOC-03,
  35th Annual ACM Symposium on Theory of computing}}}\ (\bibinfo {year}
  {2003})\ pp.\ \bibinfo {pages} {59--68},\ \Eprint
  {http://arxiv.org/abs/quant-ph/0209131} {arXiv:quant-ph/0209131} \BibitemShut
  {NoStop}%
\bibitem [{\citenamefont {Wiebe}\ \emph {et~al.}(2010)\citenamefont {Wiebe},
  \citenamefont {Berry}, \citenamefont {H{\o}yer},\ and\ \citenamefont
  {Sanders}}]{wiebe2010higher}%
  \BibitemOpen
  \bibfield  {author} {\bibinfo {author} {\bibfnamefont {N.}~\bibnamefont
  {Wiebe}}, \bibinfo {author} {\bibfnamefont {D.}~\bibnamefont {Berry}},
  \bibinfo {author} {\bibfnamefont {P.}~\bibnamefont {H{\o}yer}}, \ and\
  \bibinfo {author} {\bibfnamefont {B.~C.}\ \bibnamefont {Sanders}},\ }\href
  {\doibase 10.1088/1751-8113/43/6/065203} {\bibfield  {journal} {\bibinfo
  {journal} {J. Phys. A: Math. Theor.}\ }\textbf {\bibinfo {volume} {43}},\
  \bibinfo {pages} {065203} (\bibinfo {year} {2010})},\ \Eprint
  {http://arxiv.org/abs/0812.0562} {arXiv:0812.0562} \BibitemShut {NoStop}%
\bibitem [{\citenamefont {Wang}\ \emph {et~al.}(2018)\citenamefont {Wang},
  \citenamefont {Li}, \citenamefont {Yin},\ and\ \citenamefont
  {Zeng}}]{wang2018ibm}%
  \BibitemOpen
  \bibfield  {author} {\bibinfo {author} {\bibfnamefont {Y.}~\bibnamefont
  {Wang}}, \bibinfo {author} {\bibfnamefont {Y.}~\bibnamefont {Li}}, \bibinfo
  {author} {\bibfnamefont {Z.-q.}\ \bibnamefont {Yin}}, \ and\ \bibinfo
  {author} {\bibfnamefont {B.}~\bibnamefont {Zeng}},\ }\href {\doibase
  10.1038/s41534-018-0095-x} {\bibfield  {journal} {\bibinfo  {journal} {npj
  Quantum Inf.}\ }\textbf {\bibinfo {volume} {4}},\ \bibinfo {pages} {46}
  (\bibinfo {year} {2018})},\ \Eprint {http://arxiv.org/abs/1801.03782}
  {arXiv:1801.03782} \BibitemShut {NoStop}%
\bibitem [{\citenamefont {Cao}\ \emph {et~al.}(2012)\citenamefont {Cao},
  \citenamefont {Daskin}, \citenamefont {Frankel},\ and\ \citenamefont
  {Kais}}]{cao2012quantum}%
  \BibitemOpen
  \bibfield  {author} {\bibinfo {author} {\bibfnamefont {Y.}~\bibnamefont
  {Cao}}, \bibinfo {author} {\bibfnamefont {A.}~\bibnamefont {Daskin}},
  \bibinfo {author} {\bibfnamefont {S.}~\bibnamefont {Frankel}}, \ and\
  \bibinfo {author} {\bibfnamefont {S.}~\bibnamefont {Kais}},\ }\href {\doibase
  10.1080/00268976.2012.668289} {\bibfield  {journal} {\bibinfo  {journal}
  {Mol. Phys.}\ }\textbf {\bibinfo {volume} {110}},\ \bibinfo {pages}
  {1675–1680} (\bibinfo {year} {2012})},\ \Eprint
  {http://arxiv.org/abs/1110.2232} {arXiv:1110.2232} \BibitemShut {NoStop}%
\bibitem [{\citenamefont {Cao}\ \emph {et~al.}(2013)\citenamefont {Cao},
  \citenamefont {Papageorgiou}, \citenamefont {Petras}, \citenamefont {Traub},\
  and\ \citenamefont {Kais}}]{cao2013quantum}%
  \BibitemOpen
  \bibfield  {author} {\bibinfo {author} {\bibfnamefont {Y.}~\bibnamefont
  {Cao}}, \bibinfo {author} {\bibfnamefont {A.}~\bibnamefont {Papageorgiou}},
  \bibinfo {author} {\bibfnamefont {I.}~\bibnamefont {Petras}}, \bibinfo
  {author} {\bibfnamefont {J.}~\bibnamefont {Traub}}, \ and\ \bibinfo {author}
  {\bibfnamefont {S.}~\bibnamefont {Kais}},\ }\href {\doibase
  10.1088/1367-2630/15/1/013021} {\bibfield  {journal} {\bibinfo  {journal}
  {New J. Phys.}\ }\textbf {\bibinfo {volume} {15}},\ \bibinfo {pages} {013021}
  (\bibinfo {year} {2013})},\ \Eprint {http://arxiv.org/abs/1207.2485}
  {arXiv:1207.2485} \BibitemShut {NoStop}%
\bibitem [{\citenamefont {Gottesman}\ and\ \citenamefont
  {Chuang}(2001)}]{gottesman2001quantum}%
  \BibitemOpen
  \bibfield  {author} {\bibinfo {author} {\bibfnamefont {D.}~\bibnamefont
  {Gottesman}}\ and\ \bibinfo {author} {\bibfnamefont {I.}~\bibnamefont
  {Chuang}},\ }\href@noop {} {\  (\bibinfo {year} {2001})},\ \Eprint
  {http://arxiv.org/abs/quant-ph/0105032} {arXiv:quant-ph/0105032} \BibitemShut
  {NoStop}%
\bibitem [{\citenamefont {Coles}\ \emph {et~al.}(2018)\citenamefont {Coles},
  \citenamefont {Eidenbenz}, \citenamefont {Pakin}, \citenamefont {Adedoyin},
  \citenamefont {Ambrosiano}, \citenamefont {Anisimov}, \citenamefont {Casper},
  \citenamefont {Chennupati}, \citenamefont {Coffrin}, \citenamefont {Djidjev},
  \citenamefont {Gunter}, \citenamefont {Karra}, \citenamefont {Lemons},
  \citenamefont {Lin}, \citenamefont {Lokhov}, \citenamefont {Malyzhenkov},
  \citenamefont {Mascarenas}, \citenamefont {Mniszewski}, \citenamefont
  {Nadiga}, \citenamefont {O'Malley}, \citenamefont {Oyen}, \citenamefont
  {Prasad}, \citenamefont {Roberts}, \citenamefont {Romero}, \citenamefont
  {Santhi}, \citenamefont {Sinitsyn}, \citenamefont {Swart}, \citenamefont
  {Vuffray}, \citenamefont {Wendelberger}, \citenamefont {Yoon}, \citenamefont
  {Zamora},\ and\ \citenamefont {Zhu}}]{coles2018quantum}%
  \BibitemOpen
  \bibfield  {author} {\bibinfo {author} {\bibfnamefont {P.~J.}\ \bibnamefont
  {Coles}}, \bibinfo {author} {\bibfnamefont {S.}~\bibnamefont {Eidenbenz}},
  \bibinfo {author} {\bibfnamefont {S.}~\bibnamefont {Pakin}}, \bibinfo
  {author} {\bibfnamefont {A.}~\bibnamefont {Adedoyin}}, \bibinfo {author}
  {\bibfnamefont {J.}~\bibnamefont {Ambrosiano}}, \bibinfo {author}
  {\bibfnamefont {P.}~\bibnamefont {Anisimov}}, \bibinfo {author}
  {\bibfnamefont {W.}~\bibnamefont {Casper}}, \bibinfo {author} {\bibfnamefont
  {G.}~\bibnamefont {Chennupati}}, \bibinfo {author} {\bibfnamefont
  {C.}~\bibnamefont {Coffrin}}, \bibinfo {author} {\bibfnamefont
  {H.}~\bibnamefont {Djidjev}}, \bibinfo {author} {\bibfnamefont
  {D.}~\bibnamefont {Gunter}}, \bibinfo {author} {\bibfnamefont
  {S.}~\bibnamefont {Karra}}, \bibinfo {author} {\bibfnamefont
  {N.}~\bibnamefont {Lemons}}, \bibinfo {author} {\bibfnamefont
  {S.}~\bibnamefont {Lin}}, \bibinfo {author} {\bibfnamefont {A.}~\bibnamefont
  {Lokhov}}, \bibinfo {author} {\bibfnamefont {A.}~\bibnamefont {Malyzhenkov}},
  \bibinfo {author} {\bibfnamefont {D.}~\bibnamefont {Mascarenas}}, \bibinfo
  {author} {\bibfnamefont {S.}~\bibnamefont {Mniszewski}}, \bibinfo {author}
  {\bibfnamefont {B.}~\bibnamefont {Nadiga}}, \bibinfo {author} {\bibfnamefont
  {D.}~\bibnamefont {O'Malley}}, \bibinfo {author} {\bibfnamefont
  {D.}~\bibnamefont {Oyen}}, \bibinfo {author} {\bibfnamefont {L.}~\bibnamefont
  {Prasad}}, \bibinfo {author} {\bibfnamefont {R.}~\bibnamefont {Roberts}},
  \bibinfo {author} {\bibfnamefont {P.}~\bibnamefont {Romero}}, \bibinfo
  {author} {\bibfnamefont {N.}~\bibnamefont {Santhi}}, \bibinfo {author}
  {\bibfnamefont {N.}~\bibnamefont {Sinitsyn}}, \bibinfo {author}
  {\bibfnamefont {P.}~\bibnamefont {Swart}}, \bibinfo {author} {\bibfnamefont
  {M.}~\bibnamefont {Vuffray}}, \bibinfo {author} {\bibfnamefont
  {J.}~\bibnamefont {Wendelberger}}, \bibinfo {author} {\bibfnamefont
  {B.}~\bibnamefont {Yoon}}, \bibinfo {author} {\bibfnamefont {R.}~\bibnamefont
  {Zamora}}, \ and\ \bibinfo {author} {\bibfnamefont {W.}~\bibnamefont {Zhu}},\
  }\href@noop {} {\  (\bibinfo {year} {2018})},\ \Eprint
  {http://arxiv.org/abs/1804.03719} {arXiv:1804.03719} \BibitemShut {NoStop}%
\bibitem [{\citenamefont {Abadi}\ \emph {et~al.}(2016)\citenamefont {Abadi},
  \citenamefont {Barham}, \citenamefont {Chen}, \citenamefont {Chen},
  \citenamefont {Davis}, \citenamefont {Dean}, \citenamefont {Devin},
  \citenamefont {Ghemawat}, \citenamefont {Irving}, \citenamefont {Isard},
  \citenamefont {Kudlur}, \citenamefont {Levenberg}, \citenamefont {Monga},
  \citenamefont {Moore}, \citenamefont {Murray}, \citenamefont {Steiner},
  \citenamefont {Tucker}, \citenamefont {Vasudevan}, \citenamefont {Warden},
  \citenamefont {Wicke}, \citenamefont {Yu},\ and\ \citenamefont
  {Zheng}}]{abadi2016tensorflow}%
  \BibitemOpen
  \bibfield  {author} {\bibinfo {author} {\bibfnamefont {M.}~\bibnamefont
  {Abadi}}, \bibinfo {author} {\bibfnamefont {P.}~\bibnamefont {Barham}},
  \bibinfo {author} {\bibfnamefont {J.}~\bibnamefont {Chen}}, \bibinfo {author}
  {\bibfnamefont {Z.}~\bibnamefont {Chen}}, \bibinfo {author} {\bibfnamefont
  {A.}~\bibnamefont {Davis}}, \bibinfo {author} {\bibfnamefont
  {J.}~\bibnamefont {Dean}}, \bibinfo {author} {\bibfnamefont {M.}~\bibnamefont
  {Devin}}, \bibinfo {author} {\bibfnamefont {S.}~\bibnamefont {Ghemawat}},
  \bibinfo {author} {\bibfnamefont {G.}~\bibnamefont {Irving}}, \bibinfo
  {author} {\bibfnamefont {M.}~\bibnamefont {Isard}}, \bibinfo {author}
  {\bibfnamefont {M.}~\bibnamefont {Kudlur}}, \bibinfo {author} {\bibfnamefont
  {J.}~\bibnamefont {Levenberg}}, \bibinfo {author} {\bibfnamefont
  {R.}~\bibnamefont {Monga}}, \bibinfo {author} {\bibfnamefont
  {S.}~\bibnamefont {Moore}}, \bibinfo {author} {\bibfnamefont {D.~G.}\
  \bibnamefont {Murray}}, \bibinfo {author} {\bibfnamefont {B.}~\bibnamefont
  {Steiner}}, \bibinfo {author} {\bibfnamefont {P.}~\bibnamefont {Tucker}},
  \bibinfo {author} {\bibfnamefont {V.}~\bibnamefont {Vasudevan}}, \bibinfo
  {author} {\bibfnamefont {P.}~\bibnamefont {Warden}}, \bibinfo {author}
  {\bibfnamefont {M.}~\bibnamefont {Wicke}}, \bibinfo {author} {\bibfnamefont
  {Y.}~\bibnamefont {Yu}}, \ and\ \bibinfo {author} {\bibfnamefont
  {X.}~\bibnamefont {Zheng}},\ }in\ \href@noop {} {\emph {\bibinfo {booktitle}
  {Proceedings of the 12th USENIX Conference on Operating Systems Design and
  Implementation}}}\ (\bibinfo {year} {2016})\ p.\ \bibinfo {pages}
  {265–283},\ \Eprint {http://arxiv.org/abs/1605.08695} {arXiv:1605.08695}
  \BibitemShut {NoStop}%
\bibitem [{\citenamefont {Paszke}\ \emph {et~al.}(2017)\citenamefont {Paszke},
  \citenamefont {Gross}, \citenamefont {Chintala},\ and\ \citenamefont
  {Chanan}}]{paszke2017pytorch}%
  \BibitemOpen
  \bibfield  {author} {\bibinfo {author} {\bibfnamefont {A.}~\bibnamefont
  {Paszke}}, \bibinfo {author} {\bibfnamefont {S.}~\bibnamefont {Gross}},
  \bibinfo {author} {\bibfnamefont {S.}~\bibnamefont {Chintala}}, \ and\
  \bibinfo {author} {\bibfnamefont {G.}~\bibnamefont {Chanan}},\ }in\ \href
  {https://openreview.net/forum?id=BJJsrmfCZ} {\emph {\bibinfo {booktitle}
  {{Workshop Proceedings of the 31st conference on Neural Information
  Processing Systems}}}}\ (\bibinfo {year} {2017})\BibitemShut {NoStop}%
\end{thebibliography}%

\end{document}